\documentclass[twoside]{article}

\usepackage{tikz,stmaryrd,amssymb,amsmath,amsthm,url}

\usepackage{etoolbox}
\makeatletter
\patchcmd{\l@section}
  {\hfil}
  {\leaders\hbox{\normalfont$\m@th\mkern \@dotsep mu\hbox{.}\mkern \@dotsep mu$}\hfill}
  {}{}
\makeatother

\newtheorem{theorem}{Theorem}[section]
\newtheorem{lemma}[theorem]{Lemma}
\newtheorem{proposition}[theorem]{Proposition}
\newtheorem{corollary}[theorem]{Corollary}
\newtheorem{definition}[theorem]{Definition} 
\newtheorem{question}[theorem]{Question} 
\theoremstyle{definition}
\newtheorem{example}[theorem]{Example}

\newtheorem*{thm:sclcpl*}{Theorem~\ref{thm:sclcpl}}
\newtheorem*{thm:nfs*}{Theorem~\ref{thm:nfs}}

\newcommand{\ALPHA}{\alpha}
\newcommand{\Au}{A^{u}}

\newcommand{\baf}{\ensuremath{\mathit{bf}}}
\newcommand{\T}{\NT}
\newcommand{\NT}{\ensuremath{\mathcal{T}_A}}
\newcommand{\BF}{\ensuremath{\textit{BF}_A}}
\newcommand{\BFf}[1]{\ensuremath{\textit{BF}_{#1}}}
\newcommand{\PS}{\ensuremath{C_A}}
\newcommand{\PSf}[1]{\ensuremath{C_{#1}}}
\newcommand{\ci}{h}
\newcommand{\cj}{j}
\newcommand{\Fi}{H}
\newcommand{\Fj}{J}
\newcommand{\ri}{r}
\newcommand{\Le}{L}
\newcommand{\Ri}{R}
\newcommand{\f}{F}
\newcommand{\g}{G}
\newcommand{\fr}{\ensuremath{\mathit{se}}}
\newcommand{\rp}{\ensuremath{\mathit{rpf}}}
\newcommand{\rpt}{\ensuremath{\mathit{rp}}}
\newcommand{\rpse}{\ensuremath{\mathit{rpse}}}
\newcommand{\rpbf}{\ensuremath{\mathit{rpbf}}}
\newcommand{\con}{\ensuremath{\mathit{cf}}}
\newcommand{\crt}{\ensuremath{\mathit{cr}}}
\newcommand{\crse}{\ensuremath{\mathit{cse}}}
\newcommand{\crbf}{\ensuremath{\mathit{cbf}}}
\newcommand{\mem}{\ensuremath{\mathit{mf}}}
\newcommand{\memt}{\ensuremath{\mathit{m}}}
\newcommand{\memse}{\ensuremath{\mathit{mse}}}
\newcommand{\membf}{\ensuremath{\mathit{mbf}}}
\newcommand{\stse}{\ensuremath{\mathit{sse}}}
\newcommand{\stbf}{\ensuremath{\mathit{sbf}}}
\newcommand{\vcstse}[1]{\ensuremath{=_{\stse,#1}}}
\newcommand{\vcstbf}[1]{\ensuremath{=_{\stbf,#1}}}
\newcommand{\tr}{\ensuremath{{\sf T}}}
\newcommand{\fa}{\ensuremath{{\sf F}}}
\newcommand{\true}{\ensuremath{\textit{true}}}
\newcommand{\false}{\ensuremath{\textit{false}}}

\newcommand{\lef}{{\raisebox{0pt}{\footnotesize$\;\triangleleft~$}}}
\newcommand{\rig}{{\raisebox{0pt}{\footnotesize$~\triangleright\;$}}}
\renewcommand{\unlhd}{{\raisebox{1pt}{\footnotesize$\;\underline{\triangleleft}~$}}}
\renewcommand{\unrhd}{{\raisebox{1pt}{\footnotesize$~\underline{\triangleright}\;$}}}
\newcommand{\CP}{\ensuremath{\normalfont\axname{CP}}}
\newcommand{\CPrp}{\CP_{\mathit{rp}}(A)}
\newcommand{\CPcon}{\CPcr}
\newcommand{\CPcr}{\CP_{\mathit{cr}}(A)}
\newcommand{\CPmem}{\CP_{\mathit{mem}}}
\newcommand{\CPstat}{\CP_{\mathit{s}}}
\newcommand{\CPst}{\CP_{\mathit{st}}}

\newcommand{\axname}[1]{\textup{{\ensuremath{\textrm{#1}}}}}

\newcommand{\qedex}{\textit{End example.}}

\begin{document}

\title{Evaluation Trees for Proposition Algebra}

\author{%
    Jan A. Bergstra \quad\&\quad Alban Ponse\\[1mm]
 {\small Section Theory of Computer Science, Informatics Institute}\\[-1mm]
 {\small  Faculty of Science, University of Amsterdam}\\
 {\small \url{https://staff.science.uva.nl/{j.a.bergstra/,a.ponse/}}}
}

\date{}

\maketitle

\begin{abstract}
Proposition algebra is based on Hoare's 
conditional connective, which is a ternary connective comparable to if-then-else and used in 
the setting of propositional logic.
Conditional statements are provided with a simple semantics that is based on evaluation 
trees and that characterizes so-called free valuation congruence: 
two conditional statements are free valuation 
congruent if, and only if, they have equal evaluation trees.
Free valuation congruence is axiomatized by the four basic equational axioms of 
proposition algebra that define the conditional connective.
Valuation congruences that identify more conditional statements
than free valuation congruence are repetition-proof,
contractive, memorizing, and static valuation congruence. 
Each of these valuation congruences is characterized using a transformation 
on evaluation trees:
two conditional statements are C-valuation 
congruent if, and only if, their C-transformed evaluation trees are equal.
These transformations are simple and natural, and only for static valuation 
congruence a slightly more complex transformation is used. 
Also, each of these valuation congruences is axiomatized in proposition algebra.
A spin-off of our
approach can be called ``normalization functions for proposition algebra'': 
for each valuation congruence C considered, 
two conditional statements are C-valuation congruent if, and only if, 
the C-normalization function returns equal images.
\\[2mm]
\end{abstract}

{\small \tableofcontents}
\section{Introduction}
\label{sec:Intro}
In 1985, Hoare's paper \emph{A couple of novelties in the propositional calculus}~\cite{Hoa85}
was published. In this paper the ternary connective
\(\_\lef \_\rig \_\)
is introduced as the \emph{conditional}.\footnote{To be 
 distinguished from Hoare's \emph{conditional}
 introduced in his 1985~book on CSP~\cite{Hoa85a} 
 and in his well-known 1987 paper \emph{Laws of Programming}
 \cite{HHH87} 
 for expressions $P\lef b\rig Q$ with $P$ and $Q$ denoting programs and 
 $b$ a Boolean expression.}
A more common expression for a conditional statement 
\[P\lef Q\rig R\]
is
\[\texttt{if }Q\texttt{ then }P\texttt{ else }R,\]
but in order to reason 
systematically with conditional statements, a notation
such as $P\lef Q\rig R$ is preferable. 
In a conditional statement $P\lef Q\rig R$, first $Q$ is evaluated, and depending 
on that evaluation result, 
either $P$ or $R$ is evaluated (and the other is not) and determines
the evaluation value. 
This evaluation strategy is a form of \emph{short-circuit} evaluation.\footnote{Short-circuit 
   evaluation denotes the semantics of binary propositional connectives in which the second
   argument is evaluated only if the first argument does not suffice 
   to determine the value of the expression.}
In \cite{Hoa85}, Hoare proves that propositional logic 
is characterized by eleven equational axioms, some of which employ constants 
$\tr$ and \fa\ for the truth values \true\ and \false. 

In 2011, we introduced \emph{Proposition Algebra} in~\cite{BP10} as a general approach to the 
study of the conditional: we defined several \emph{valuation congruences} and provided
equational axiomatizations of these congruences. 
The most basic and least identifying valuation congruence is \emph{free} valuation congruence, 
which is axiomatized by the axioms in Table~\ref{CP}.
\begin{table}[t]
\centering
\hrule
\begin{align*}
\label{cp1}
\tag{CP1} x \lef \tr \rig y &= x\\
\label{cp2}\tag{CP2}
x \lef \fa \rig y &= y\\
\label{cp3}\tag{CP3}
\tr \lef x \rig \fa  &= x\\
\label{cp4}\tag{CP4}
\qquad
    x \lef (y \lef z \rig u)\rig v &= 
	(x \lef y \rig v) \lef z \rig (x \lef u \rig v)
\end{align*}
\hrule
\label{CP}
\caption{The set \CP\ of equational axioms for free valuation congruence}
\end{table}
These axioms stem from~\cite{Hoa85} and define the conditional as a primitive connective.
We use the name \CP\ (for Conditional Propositions) for this set of axioms.
Interpreting a conditional statement as an if-then-else expression, 
axioms~\eqref{cp1}-\eqref{cp3} are natural, and axiom~\eqref{cp4} 
(distributivity) can be clarified by case analysis: 
if $z$ evaluates to \true\ and $y$ as well, then  $x$ determines the
result of evaluation;
if $z$ evaluates to \true\ and $y$ evaluates to \false, then  $v$ determines the
result of evaluation, and so on and so forth.
A simple example, taken from~\cite{BP10}, is the conditional statement that a pedestrian 
evaluates before crossing a road with two-way traffic driving on the right:
\[
(\textit{look-left-and-check}\:\lef\:\textit{look-right-and-check}\:\rig\fa)
\lef\:\textit{look-left-and-check}\:\rig\fa.
\]
This statement requires one, or two, or three atomic evaluations and cannot be
simplified to one that requires less.\footnote{\label{fn1}Note that 
$\textit{look-left-and-check}\:\lef(\textit{look-right-and-check}\:\lef\:
\textit{look-left-and-check}\:\rig\fa)\rig\fa$ prescribes by axioms~\eqref{cp4} and \eqref{cp2}
the same evaluation.}

In \S\ref{sec:free} we characterize free valuation congruence
with help of \emph{evaluation trees}: 
given a conditional statement, its evaluation tree 
directly represents all its evaluations (in the way a truth table does in
the case of propositional logic). 
Two conditional statements are equivalent with respect to free valuation congruence
if their evaluation trees are equal.
Evaluation trees are simple binary trees, 
proposed by Daan Staudt in~\cite{Daan} (that appeared in 2012). 
Free valuation congruence identifies less than the equivalence
defined by Hoare's axioms in~\cite{Hoa85}.
For example, the atomic
proposition $a$ and  the conditional statement $\tr\lef a\rig a$
are not equivalent with respect to free valuation congruence,
although they are  equivalent with respect to \emph{static} valuation congruence, 
which is the valuation congruence that characterizes propositional logic.

In \S\ref{sec:rp} we consider \emph{repetition-proof} valuation congruence,
a valuation congruence that identifies more than free and less than static 
valuation congruence. Repetition-proof valuation congruence is axiomatized 
by \CP\ extended with two (schematic) axioms, one of which reads 
\[x\lef a\rig (y\lef a\rig z)=x\lef a\rig (z\lef a\rig z)\]
for any atomic proposition $a$,
and thus expresses that if $a$ evaluates to \false, a consecutive evaluation
of $a$ also evaluates to \false, so the conditional statement at the $y$-position will not 
be evaluated and can be replaced by any other.
As an example, 
\[\tr\lef a\rig a=\tr\lef a\rig(\tr\lef a\rig\fa)=\tr\lef a\rig(\fa\lef a\rig\fa),\]
and the left-hand and right-hand conditional statements 
are equivalent with respect to repetition-proof valuation congruence,
but not with respect to free valuation congruence.
We characterize repetition-proof valuation congruence by 
defining a transformation on evaluation 
trees that yields \emph{repetition-proof evaluation trees}: two conditional
statements are equivalent with respect to repetition-proof valuation congruence
if, and only if, they have equal {repetition-proof evaluation trees}. 
Although this transformation on evaluation trees is simple and 
natural, our proof of the mentioned characterization 
 |which is phrased as a completeness result| 
is non-trivial and we could not find a proof that is essentially simpler. 

Valuation congruences that identify more conditional statements
than repetition-proof valuation congruence are 
contractive, memorizing, and static valuation congruence, and these are all
defined and axiomatized in~\cite{BP10}.
In \S\S\ref{sec:cr}-\ref{sec:stat}, each of these valuation congruences is 
characterized using a transformation on evaluation trees:
two conditional statements are C-valuation 
congruent if, and only if, their C-transformed evaluation trees are equal.
These transformations are simple and natural, and only for static valuation 
congruence we use a slightly more complex transformation. 

In \S\ref{sec:Conc} we discuss the general structure of the proofs of the
axiomatization results in \S\S\ref{sec:rp}-\ref{sec:stat}, each of which 
is based on a normalization function for conditional statements.
Then we end the paper with a brief digression on short-circuit logic,
an example on the use of repetition-proof valuation congruence, and some
remarks about side effects. 

\medskip

A shortened version of this paper, not covering the material in 
\S\S\ref{sec:cr}-\ref{sec:stat}, was published as~\cite{BP12a}.
However, we simplified the proof of Lemma 3.15 in~\cite{BP12a} (Lemma~\ref{la:3.14} 
in this version).

\section{Evaluation trees for free valuation congruence}
\label{sec:free}
Consider the signature
\(\Sigma_{\CP}(A)=\{\_\lef\_\rig\_\,,\tr,\fa,a\mid a\in A\}\)
with constants \tr\ and \fa\
for the truth values \true\ and \false, respectively, and constants $a$
for atomic propositions, further called \emph{atoms}, from some countable 
set $A$ containing at least two atoms. We write
\[\PS\]
for the set of closed terms, or \emph{conditional statements}, 
over the signature $\Sigma_{\CP}(A)$. Given a conditional  
statement $P\lef Q\rig R$, we  refer
to $Q$ as its \emph{central condition}.

We define the \emph{dual} $P^d$ of $P\in\PS$ as follows:
\begin{align*}
\tr^d&=\fa,
&a^d&=a \quad(\text{for $a\in A$}),\\
\fa^d&=\tr,
&(P\lef Q\rig R)^d&=R^d\lef Q^d\rig P^d.
\end{align*} 
Observe that \CP\ is a self-dual axiomatization:
when defining $x^d=x$ for each variable $x$, the dual of each axiom is also in \CP,
and hence
\[\CP\vdash P=Q ~\iff~ \CP\vdash P^d=Q^d.\]

A natural view on conditional statements in \PS\
involves short-circuit evaluation, similar to how we 
consider the evaluation of an
``$\texttt{if }y\texttt{ then }x\texttt{ else }z$'' expression.
The following definition is taken from~\cite{Daan}.

\begin{definition}
\label{def:trees}
The set \NT\ of \textbf{evaluation trees over $A$ with leaves in 
$\{\tr, \fa\}$} is defined inductively by
\begin{align*}
\tr&\in\NT,\\
\fa&\in\NT,\\
(X\unlhd a\unrhd Y)&\in\NT 
\text{ for any }X,Y \in \NT \text{  and } a\in A.
\end{align*}
The function $\_\unlhd a\unrhd\_$ is called 
\textbf{post-conditional composition over $a$}.
In the evaluation tree $X \unlhd a \unrhd Y$, 
the root is represented by $a$,
the left branch by $X$ and the right branch by $Y$. 
\end{definition}

We refer to trees in \NT\ as evaluation trees, or trees for short. 
Post-conditional composition and its notation stem from~\cite{BL}.
Evaluation trees play a crucial role in the 
main results of~\cite{Daan}.
In order to define our ``evaluation tree semantics'', we first define 
an auxiliary function on trees.
\begin{definition} [Leaf replacement]
Given evaluation trees $Y,Z\in\NT$, the \textbf{leaf replacement function}
$[\tr\mapsto Y, \fa \mapsto Z]:\NT\to\NT$ for which postfix notation 
\[X[\tr\mapsto Y, \fa \mapsto Z]\]
is adopted, is defined as follows, where $a \in A$:
\begin{align*}
\tr[\tr\mapsto Y,\fa\mapsto Z]&= Y,\\
\fa[\tr\mapsto Y,\fa\mapsto Z]&= Z,\\
(X_1\unlhd a\unrhd X_2)[\tr\mapsto Y,\fa\mapsto Z]
&=X_1[\tr\mapsto Y,\fa\mapsto Z]\unlhd a\unrhd X_2[\tr\mapsto Y,\fa\mapsto Z].
\end{align*}
\end{definition}
We note that the order in which the replacements of leaves of 
$X$ is listed
is irrelevant and adopt the convention of not listing  
identities inside the brackets, e.g., 
$X[\fa\mapsto Z]=X[\tr\mapsto \tr,\fa\mapsto Z]$.
Furthermore, repeated replacements satisfy the following equation:
\begin{align*}
\big(X[\tr\mapsto Y_1,&\fa\mapsto Z_1]\big)\;[\tr\mapsto Y_2,\fa\mapsto Z_2]\\
&=
X[\tr\mapsto Y_1[\tr\mapsto Y_2,\fa\mapsto Z_2],~
\fa\mapsto Z_1[\tr\mapsto Y_2,\fa\mapsto Z_2]].
\end{align*}

We now have the terminology and notation to 
define the interpretation of conditional statements in \PS\ as evaluation trees
by a function $se$ (abbreviating  short-circuit evaluation).

\begin{definition}
\label{def:se}
The \textbf{short-circuit evaluation function} $se : \PS \to\NT$ 
is defined as
follows, where $a\in A$:
\begin{align*}
se(\tr)&=\tr,\\
se(\fa)&=\fa,\\
se(a)&=\tr\unlhd a\unrhd \fa,\\
se(P \lef Q\rig R)&= se(Q)[\tr\mapsto se(P), \fa\mapsto se(R)].
\end{align*}
\end{definition}

\begin{example}\rm
\label{ex:fr}
The conditional statement $a\lef(\fa\lef a\rig\tr) \rig \fa$ yields 
the following evaluation tree:
\begin{align*}
se(a\lef(\fa\lef a\rig\tr) \rig \fa)
&=se(\fa\lef a\rig\tr)[\tr\mapsto se(a), \fa\mapsto se(\fa)]\\
&=(\fa\unlhd a\unrhd\tr)[\tr\mapsto se(a)]\\
&=\fa\unlhd a\unrhd(\tr\unlhd a\unrhd \fa).
\end{align*}
A more pictorial representation of this evaluation tree is the following, 
where ${\unlhd}$ yields a left branch and ${\unrhd}$ a right branch:
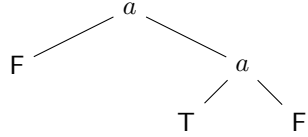
\[
\begin{tikzpicture}[%
      level distance=7.5mm,
      level 1/.style={sibling distance=30mm},
      level 2/.style={sibling distance=15mm},
      baseline=(current bounding box.center)]
      \node (a) {$a$}
        child {node (b1) {$\fa$}
        }
        child {node (b2) {$a$}
          child {node (d1) {$\tr$}} 
          child {node (d2) {$\fa$}}
        };
\end{tikzpicture}
\]
\qedex\end{example}

As we can see from the definition on atoms, evaluation 
continues in the left branch if an atom evaluates to \true\ and in 
the right branch if it evaluates to \false.
We shall often use the constants \tr\ and \fa\ to denote the result 
of an evaluation (instead of \true\ and \false).

\begin{definition}
\label{def:eval}
Let $P\in\PS$. An \textbf{evaluation} of $P$ is a pair 
\((\sigma, B)\)
where $\sigma\in(A\{\tr,\fa\})^*$ and $B\in\{\tr,\fa\}$,
such that if $se(P)\in\{\tr,\fa\}$, then $\sigma=\epsilon$ (the empty string)
and $B=se(P)$, and 
otherwise, 
\[\sigma=a_1B_1a_2B_2\cdots a_nB_n,\]
where $a_1a_2\cdots a_nB$ is a complete path in $se(P)$ 
and
\begin{itemize}\setlength\itemsep{-.1em}
\item[$-$]
for $i<n$, if $a_{i+1}$ is a left child of $a_i$ then $B_i=\tr$, and otherwise
$B_i=\fa$,
\item[$-$]
if $B$ is a left child of $a_n$ then $B_n=\tr$, and otherwise $B_n=\fa$.
\end{itemize}
We refer to $\sigma$ as the \textbf{evaluation path} and
to $B$ as the \textbf{evaluation result}. 
\end{definition}

So, an evaluation of a conditional statement $P$ is a complete path in $se(P)$
(from root to leaf) and contains evaluation values for all occurring atoms. 
For instance, the evaluation tree $\fa\unlhd a\unrhd(\tr\unlhd a\unrhd \fa)$
from Example~\ref{ex:fr} encodes the evaluations
$(a\tr,\fa)$, $(a\fa a\tr,\tr)$, and $(a\fa a\fa,\fa)$.
As an aside, we note that this particular evaluation tree encodes all
possible evaluations of $\neg a\;\texttt{\&\&}\;a$, where \texttt{\&\&}
is the connective that prescribes \emph{short-circuit conjunction}
(we return to this connective in \S\ref{sec:Conc}).

In turn, each evaluation tree gives rise to a \emph{unique}
conditional statement. For Example~\ref{ex:fr}, this is
$\fa\lef a\rig (\tr\lef a\rig\fa)$ (note the syntactical correspondence).

\begin{definition}
\label{def:basic}
\textbf{Basic forms over $A$} are defined by the following grammar
\[t::= \tr\mid\fa\mid t\lef a \rig t\quad\text{for $a\in A$.}\]
We write $\BF$ for the set of basic forms over $A$. 
The \textbf{depth} $d(P)$ of $P\in\BF$ is defined by 
$d(\tr) = d(\fa) = 0$ and 
$d(Q \lef a \rig R) = 1 + \max\{d(Q), d(R)\}$. 
\end{definition}

The following lemma's exploit the structure of basic forms and
are stepping stones to our first completeness result (Theorem~\ref{thm:1}).

\begin{lemma}
\label{la:2.7}
For each $P\in\PS$ there exists $Q\in\BF$ such that
$\CP\vdash P=Q$.
\end{lemma}

\begin{proof} First we establish an auxiliary result:
 if $P,Q,R$ are basic forms, then there is a basic form $S$ such that
$\CP\vdash P\lef Q\rig R=S$. This follows by structural induction on $Q$.

\smallskip

The lemma's statement follows by structural induction on $P$. The base
cases $P\in\{\tr,\fa,a\mid a\in A\}$ are trivial, and if $P= P_1\lef P_2\rig
P_3$ there exist by induction basic forms $Q_i$ such that $\CP\vdash P_i=Q_i$, hence
$\CP\vdash P_1\lef P_2\rig P_3=Q_1\lef Q_2\rig Q_3$. Now apply the auxiliary result.
\end{proof}

\begin{lemma}
\label{la:2.8}
For all basic forms $P$ and $Q$, $se(P)=se(Q)$ implies $P= Q$.
\end{lemma}

\begin{proof} By structural induction on $P$. 
The base cases $P\in\{\tr,\fa\}$ are trivial.
If $P= P_1\lef a\rig P_2$, then $Q\not\in\{\tr,\fa\}$
and $Q\ne Q_1\lef b\rig Q_2$ with $b\ne a$, so $Q= Q_1\lef a\rig Q_2$
and $se(P_i)=se(Q_i)$. By induction we find $P_i= Q_i$, and hence $P= Q$.
\end{proof}

\begin{definition}
\label{def:freevc}
\textbf{Free valuation congruence}, notation $=_\fr$, 
is defined on \PS\ as follows: 
\[P=_\fr Q~\iff~se(P)=se(Q).\]
\end{definition}

\begin{lemma}
\label{la:2.10}
Free valuation congruence is a congruence relation.
\end{lemma}

\begin{proof}
Assume $P=_\fr P'$, $Q=_\fr Q'$, and $R=_\fr R'$. 
Then $se(P\lef Q\rig R)=se(Q)[\tr\mapsto se(P),\fa\mapsto se(R)]=
se(Q')[\tr\mapsto se(P'),\fa\mapsto se(R')]=se(P'\lef Q'\rig R')$, so
$P\lef Q\rig R=_\fr P'\lef Q'\rig R'$. 
\end{proof}

\begin{theorem}[Completeness of $\CP$]
\label{thm:1}
For all $P,Q\in\PS$, 
\[\CP\vdash P=Q~\iff~ P=_\fr Q.\]
\end{theorem}
\begin{proof}
($\Rightarrow$) By Lemma~\ref{la:2.10}, 
$=_\fr$ is a congruence relation and it easily follows
that all \CP-axioms are sound. For example, 
soundness of axiom~\eqref{cp4} follows from
\begin{align*}
se(P\lef&(Q\lef R\rig S)\rig U)\\
&=se(Q\lef R\rig S)[\tr\mapsto se(P),\fa\mapsto se(U)]\\
&=\big(se(R)[\tr\mapsto se(Q),\fa\mapsto se(S)]\big)\;[\tr\mapsto se(P),\fa\mapsto se(U)]\\
&=se(R)[\tr\mapsto se(Q)[\tr\mapsto se(P),\fa\mapsto se(U)],
\fa\mapsto \,se(S)[\tr\mapsto se(P),\fa\mapsto se(U)]]\\
&=se(R)[\tr\mapsto se(P\lef Q\rig U),\fa\mapsto se(P\lef S\rig U)]\\
&=se((P\lef Q\rig U)\lef R\rig(P\lef S\rig U)).
\end{align*}

($\Leftarrow$) Let $P=_\fr Q$. According to Lemma~\ref{la:2.7}
there exist basic forms
$P'$ and $Q'$ such that $\CP\vdash P=P'$ and $\CP\vdash Q=Q'$, so
$\CP\vdash P'=Q'$. By soundness ($\Rightarrow$) we find $P'=_\fr Q'$, so
by Lemma~\ref{la:2.8}, $P'= Q'$. Hence, $\CP\vdash P=P'= Q'=Q$.
\end{proof}

A consequence of the above results is that for each $P\in\PS$ there is a
\emph{unique} basic form $P'$ with $\CP\vdash P=P'$, 
and that for each basic form, its $se$-image has exactly the same 
syntactic structure (replacing ${\lef}$ by ${\unlhd}$, and ${\rig}$ by ${\unrhd}$).
In the remainder of this section, we make this precise.

\begin{definition}
\label{def:bf}
The \textbf{basic form function} $\baf:\PS \to\BF$ 
is defined as follows, where $a\in A$:
\begin{align*}
\baf(\tr)&= \tr,\\
\baf(\fa)&= \fa,\\
\baf(a)&=\tr\lef a\rig \fa,\\
\baf(P \lef Q\rig R)&= \baf(Q)[\tr\mapsto \baf(P), \fa\mapsto \baf(R)].
\end{align*}
Given $Q,R\in\BF$, the auxiliary function $[\tr\mapsto Q, \fa\mapsto R]:\BF\to\BF$ 
for which post-fix 
notation $P[\tr\mapsto Q, \fa\mapsto R]$ is adopted,
is defined as follows:
\begin{align*}
\tr[\tr\mapsto Q, \fa\mapsto R]&=Q,\\
\fa[\tr\mapsto Q, \fa\mapsto R]&=R,\\
(P_1\lef a\rig P_2)[\tr\mapsto Q, \fa\mapsto R]
&=P_1[\tr\mapsto Q, \fa\mapsto R]\lef a\rig P_2[\tr\mapsto Q, \fa\mapsto R].
\end{align*}
(The notational overloading with the leaf replacement functions on evaluation trees 
is harmless).
\end{definition}

So, for given $Q,R\in \BF$, the auxiliary function $[\tr\mapsto Q,\fa\mapsto R]$
applied to $P\in\BF$ (thus, $P[\tr\mapsto Q,\fa\mapsto R]$)
replaces all \tr-occurrences in $P$ by $Q$, and all \fa-occurrences 
in $P$ by $R$. 
The following two lemma's imply that $\baf$ is a \emph{normalization function}.
 
\begin{lemma}
\label{la:2.13}
For all $P\in\PS$, $\baf(P)$ is a basic form.
\end{lemma}

\begin{proof}
By structural induction. The base cases are trivial. For the inductive case we find 
$\baf(P\lef Q\rig R)=\baf(Q)[\tr\mapsto \baf(P), \fa\mapsto \baf(R)]$, so by induction,
$\baf(P)$, $\baf(Q)$, and $\baf(R)$ are basic forms. 
Furthermore, replacing all \tr-occurrences and \fa-occurrences in $\baf(Q)$ by basic forms 
$\baf(P)$ and $\baf(R)$, respectively, yields a basic form.
\end{proof}

\begin{lemma}
\label{la:2.14}
For each basic form $P$, $\baf(P)=P$.
\end{lemma}

\begin{proof}
By structural induction on $P$.
\end{proof}

\begin{definition}
\label{def:freevca}
The binary relation $=_\baf$ on \PS\
is defined as follows: 
\[P=_\baf Q~\iff~\baf(P)=\baf(Q).\]
\end{definition}

\begin{lemma}
\label{la:2.16}
The relation $=_\baf$ is a congruence relation.
\end{lemma}

\begin{proof}
Assume $P=_\baf P'$, $Q=_\baf Q'$, and $R=_\baf R'$. 
Then $\baf(P\lef Q\rig R)=\baf(Q)[\tr\mapsto \baf(P),\fa\mapsto \baf(R)]=
\baf(Q')[\tr\mapsto \baf(P'),\fa\mapsto \baf(R')]=\baf(P'\lef Q'\rig R')$, so
$P\lef Q\rig R=_\baf P'\lef Q'\rig R'$. 
\end{proof}

Before proving that \CP\ is an axiomatization of the relation $=_\baf$,
we  show that each instance of the axiom~\eqref{cp4}
satisfies $=_\baf$.

\begin{lemma}
\label{la:2.17}
For all $P,P_1,P_2, Q_1, Q_2\in\PS$,
\begin{align*}
\baf(Q_1\lef(P_1\lef P\rig P_2)\rig Q_2)&
=\baf((Q_1\lef P_1\rig Q_2)\lef P\rig(Q_1\lef P_2\rig Q_2)).
\end{align*}
\end{lemma}

\begin{proof}
By definition, the lemma's statement is equivalent with 
\begin{align}
\nonumber
\big(\baf(P)
&[\tr\mapsto \baf(P_1),\fa\mapsto\baf(P_2)]\big)\;[\tr\mapsto\baf(Q_1),\fa\mapsto\baf(Q_2)]\\
\label{eq:bf}
&=\baf(P)[\tr\mapsto\baf(Q_1\lef P_1\rig Q_2),\fa\mapsto\baf(Q_1\lef P_2\rig Q_2)].
\end{align}
By Lemma~\ref{la:2.13}, $\baf(P)$, $\baf(P_i)$,and $\baf(Q_i)$ are basic forms.
We prove~\eqref{eq:bf} by structural induction on the form that $\baf(P)$ can have.
If $\baf(P)=\tr$, then 
\begin{align*}
\big(\tr[\tr\mapsto \baf(P_1),\fa\mapsto\baf(P_2)]\big)
&\;[\tr\mapsto\baf(Q_1),\fa\mapsto\baf(Q_2)]
=\baf(P_1)[\tr\mapsto\baf(Q_1),\fa\mapsto\baf(Q_2)] 
\end{align*}
and 
\begin{align*}
\tr[\tr\mapsto\baf(Q_1\lef P_1\rig Q_2),\fa\mapsto\baf(Q_1\lef P_2\rig Q_2)]
&=\baf(Q_1\lef P_1\rig Q_2)\\
&=\baf(P_1)[\tr\mapsto\baf(Q_1),\fa\mapsto\baf(Q_2)].
\end{align*}
If $\baf(P)=\fa$, then~\eqref{eq:bf} follows in a similar way. 

The inductive case $\baf(P)= R_1\lef a \rig R_2$ is trivial (by definition of 
the last defining clause of 
the auxiliary functions $[\tr\mapsto Q,\fa\mapsto R]$
in Definition~\ref{def:bf}). 
\end{proof}

\begin{theorem}
\label{thm:1a}
For all $P,Q\in\PS$, 
\(\CP\vdash P=Q~\iff~ P=_\baf Q.\)
\end{theorem}

\begin{proof}
($\Rightarrow$) By Lemma~\ref{la:2.16}, 
$=_\baf$ is a congruence relation and it easily follows
that closed instances of the \CP-axioms \eqref{cp1}-\eqref{cp3}
satisfy $=_\baf$. By Lemma~\ref{la:2.17} it follows that closed instances of 
axiom~\eqref{cp4} also satisfy $=_\baf$.

\smallskip

($\Leftarrow$) Assume $P=_\baf Q$. According to Lemma~\ref{la:2.7},
there exist basic forms
$P'$ and $Q'$ such that $\CP\vdash P=P'$ and $\CP\vdash Q=Q'$, so
$\CP\vdash P'=Q'$. By $\Rightarrow$ it follows that $P'=_\baf Q'$, which implies
by Lemma~\ref{la:2.14} that $P'= Q'$. Hence, $\CP\vdash P=P'= Q'=Q$.
\end{proof}

\begin{corollary}
\label{cor:1}
For all $P\in\PS$,~ $P=_\baf\baf(P)$ ~and~ $P=_\fr\baf(P)$.
\end{corollary}

\begin{proof}
By Lemma~\ref{la:2.13} and Lemma~\ref{la:2.14}, $\baf(P)=\baf(\baf(P))$, 
thus $P=_\baf \baf(P)$. By Theorem~\ref{thm:1a}, $\CP\vdash P=\baf(P)$, and
by Theorem~\ref{thm:1}, $P=_\fr\baf(P)$.
\end{proof}

\section{Evaluation trees for repetition-proof valuation congruence}
\label{sec:rp}
In~\cite{BP10} we defined \emph{repetition-proof} \CP\ as the 
extension of the axiom set \CP\ with the following two axiom schemes,
where $a$ ranges over $A$: 
\begin{align*}
\label{CPrp1}
\tag{CPrp1}
(x\lef a\rig y)\lef a\rig z&=(x\lef a\rig x)\lef a\rig z,\\
\label{CPrp2}
\tag{CPrp2}
x\lef a\rig (y\lef a\rig z)&=x\lef a\rig (z\lef a \rig z).
\end{align*}
We write
\(\CPrp\)
for this extension. These axiom schemes characterize that for each atom $a$, 
a consecutive evaluation of $a$ yields the same result, so in both cases
the conditional statement at the $y$-position will not be evaluated and can 
be replaced by any other. Note that~\eqref{CPrp1}
and \eqref{CPrp2} are each others dual. 

We define a proper subset of basic forms with the property that each
propositional statement can be proved equal to such a basic form.

\begin{definition}
\label{def:3.1}
\textbf{Rp-basic forms} are inductively defined:
\begin{itemize}\setlength\itemsep{-.1em}
\item 
\tr\ and \fa\ are rp-basic forms, and
\item 
 $P_1\lef a\rig P_2$ is an rp-basic form if $P_1$
 and $P_2$ are rp-basic forms, and if $P_i$ is not equal to \tr\ or \fa,
 then either the central condition in $P_i$ is different from $a$, or $P_i$ is of the form
 $Q_i\lef a\rig Q_i$. 
\end{itemize}
\end{definition}

It will turn out useful to define a function that transforms conditional statements
into rp-basic forms,
and that is comparable to the function $\baf$.

\begin{definition}
\label{def:3.2}
The \textbf{rp-basic form function}
$\rpbf:\PS\to \PS$
is defined by
\begin{align*}
\rpbf(P)&=\rp(\baf(P)).
\end{align*}
The auxiliary function $\rp:\BF\to\BF$ is defined as follows:
\begin{align*}
\rp(\tr)&=\tr,\\
\rp(\fa)&=\fa,\\
\rp(P\lef a\rig Q)&=\rp(f_a(P))\lef a\rig\rp(g_a(Q)).
\end{align*}
For $a\in A$, the auxiliary functions
$f_a: \BF\to\BF$ and $g_a: \BF\to\BF$ 
are defined by
\begin{align*}
f_a(\tr)&=\tr,\\
f_a(\fa)&=\fa,\\
f_a(P\lef b\rig Q)&=\begin{cases}
f_a(P)\lef a\rig  f_a(P)&\text{if } b=a,\\
P\lef b\rig Q&\text{otherwise},
\end{cases}
\end{align*}
and
\begin{align*}
g_a(\tr)&=\tr,\\
g_a(\fa)&=\fa,\\
g_a(P\lef b\rig Q)&=\begin{cases}
g_a(Q)\lef a\rig  g_a(Q)&\text{if } b=a,\\
P\lef b\rig Q&\text{otherwise}.
\end{cases}
\end{align*}
\end{definition}

Thus, $\rpbf$ maps a conditional statement $P$ to $\baf(P)$ and then
transforms $\baf(P)$ according to the auxiliary functions $\rp$, $f_a$, and $g_a$.

\begin{lemma}
\label{la:3.3}
For all $a\in A$ and $P\in\BF$, 
\[\text{$g_a(f_a(P))=f_a(f_a(P))=f_a(P)$
and  $f_a(g_a(P))=g_a(g_a(P))=g_a(P)$.}\]
\end{lemma}

\begin{proof}
By structural induction on $P$.  
The base cases $P\in\{\tr,\fa\}$
are trivial. For the inductive case $P= Q\lef b\rig R$ we have to distinguish
the cases $b=a$ and $b\ne a$. If $b=a$, then 
\begin{align*}
g_a(f_a(Q\lef a\rig R))
&=g_a(f_a(Q))\lef a\rig g_a(f_a(Q))\\
&=f_a(Q)\lef a\rig f_a(Q)
&&\text{by IH}\\
&=f_a(Q\lef a\rig R),
\end{align*}
and $f_a(f_a(Q\lef a\rig R))=f_a(Q\lef a\rig R)$ follows in a similar way.
If $b\ne a$, then $f_a(P)=g_a(P)=P$, and hence $g_a(f_a(P))=f_a(f_a(P))=f_a(P)$.

\smallskip

The second pair of equalities can be derived in a similar way.
\end{proof}

In order to prove that for all $P\in\PS$, $\rpbf(P)$ is an rp-basic form,
we use the following auxiliary lemma.

\newpage

\begin{lemma}
\label{la:3.4} 
For all $a\in A$ and $P\in\BF$, $d(P)\geq d(f_a(P))$ and $d(P)\geq d(g_a(P))$.
\end{lemma}

\begin{proof}
Fix some $a\in A$.
We prove these inequalities by structural induction on $P$.  
The base cases $P\in\{\tr,\fa\}$ are trivial. 
For the inductive case $P= Q\lef b\rig R$ 
we have to distinguish the cases $b=a$ and $b\ne a$. 
If $b=a$, then 
\begin{align*}
d(Q\lef a\rig R)
&=1+\max\{d(Q),d(R)\}\\
&\geq 1+ d(Q)\\
&\geq 1+ d(f_a(Q))
&&\text{by IH}\\
&=d(f_a(Q)\lef a\rig f_a(Q))\\
&=d(f_a(Q\lef a\rig R)),
\end{align*}
and $d(Q\lef a\rig R)\geq d(g_a(Q\lef a\rig R))$ follows in a similar way.

If $b\ne a$, then
$f_a(P)=g_a(P)=P$, and hence $d(P)\geq d(f_a(P))$
and $d(P)\geq d(g_a(P))$.
\end{proof}

\begin{lemma}
\label{la:3.5}
For all $P\in\PS$, $\rpbf(P)$ is an rp-basic form.
\end{lemma}

\begin{proof}
We first prove an auxiliary result:
\begin{equation}
\label{aux1}
\text{For all $P\in\BF$, $\rp(P)$ is an rp-basic form.}
\end{equation}
This follows by induction on the depth $d(P)$ of $P$. 
If $d(P)=0$, then $P\in\{\tr,\fa\}$, and hence $\rp(P)=P$ is an rp-basic form.
For the inductive case $d(P)=n+1$ it must be the case that $P= Q\lef a\rig R$.
We find
\begin{align*}
\rp(Q\lef a\rig R)
&=\rp(f_a(Q))\lef a\rig \rp(g_a(R)),
\end{align*}
which is an rp-basic form because
\begin{itemize}\setlength\itemsep{-.1em}
\item[$-$] 
 by Lemma~\ref{la:3.4}, $f_a(Q)$ and $g_a(R)$ are basic forms with depth smaller 
than or equal to $n$,
so by the induction hypothesis, $\rp(f_a(Q))$ and $\rp(g_a(R))$ are rp-basic forms,
\item[$-$] 
$\rp(f_a(Q))$ and $\rp(g_a(R))$ both satisfy the following  property: if the central
condition (if present) is $a$, then the outer arguments are equal. We show this first for 
$\rp(f_a(Q))$ by a case distinction on the form of $Q$:
\begin{enumerate}\setlength\itemsep{-.1em}
\item
If $Q\in\{\tr,\fa\}$, then $\rp(f_a(Q))=Q$, so there is nothing to prove.
\item 
If $Q= Q_1\lef a \rig Q_2$, then $f_a(Q)=f_a(Q_1)\lef a \rig f_a(Q_1)$ and thus
by Lemma~\ref{la:3.3}, 
\\
$\rp(f_a(Q))=\rp(f_a(Q_1))\lef a \rig \rp(f_a(Q_1))$.
\item
If $Q= Q_1\lef b \rig Q_2$ with $b\ne a$, then $f_a(Q)=Q_1\lef b\rig Q_2$, so 
\\
$\rp(f_a(Q))=\rp(f_b(Q_1))\lef b\rig \rp(g_b(Q_2))$ and there is nothing to prove.
\end{enumerate}
The fact that $\rp(g_a(R))$ satisfies this property follows in a similar way.
\end{itemize}
This finishes the proof of~\eqref{aux1}. 

\smallskip

The lemma's statement now follows by structural induction: 
the base cases (comprising a single atom $a$) are again trivial, 
and for the inductive case, 
\[\rpbf(P\lef Q\rig R)=\rp(\baf(P\lef Q\rig R))= \rp(S)\]
for some basic form $S$
by Lemma~\ref{la:2.13}, and by auxiliary result~\eqref{aux1}, $\rp(S)$
is an rp-basic form.
\end{proof}

The following, somewhat technical result is used in Proposition~\ref{prop:3.7} and
Lemma~\ref{la:3.8}.

\begin{lemma}
\label{la:3.6}
If $Q\lef a\rig R$ is an rp-basic form, then $Q=\rp(Q)=\rp(f_a(Q))$ 
and $R=\rp(R)=\rp(g_a(R))$.
\end{lemma}

\begin{proof}
We first prove an auxiliary result:
\begin{equation}
\label{eq:corr}
\text{If $Q\lef a\rig R$ is an rp-basic form, then $f_a(Q)=g_a(Q)$ and $f_a(R)=g_a(R)$.}
\end{equation}
We prove both equalities by simultaneous induction on the structure of $Q$ and $R$.
The base case, thus $Q,R\in\{\tr,\fa\}$,
is trivial.
If $Q= Q_1\lef a\rig Q_1$ and $R= R_1\lef a\rig R_1$, then
$Q$ and $R$ are rp-basic forms with central condition $a$, so
\begin{align*}
f_a(Q)
&= f_a(Q_1)\lef a\rig f_a(Q_1)\\
&= g_a(Q_1)\lef a\rig g_a(Q_1)
&&\text{by IH}\\
&=g_a(Q),
\end{align*}
and the equality for $R$ follows in a similar way.
If $Q= Q_1\lef a\rig Q_1$ and $R\ne R_1\lef a\rig R_1$,
then $f_a(R)=g_a(R)=R$, and the result follows as above. All remaining cases
follow in a similar way, which finishes the proof of~\eqref{eq:corr}.

\smallskip

We now prove the lemma's statement by simultaneous induction
on the structure of $Q$ and $R$.
The base case, thus $Q,R\in\{\tr,\fa\}$, is again trivial.
If $Q= Q_1\lef a\rig Q_1$ and $R= R_1\lef a\rig R_1$, then
by auxiliary result~\eqref{eq:corr},
\begin{align*}
\rp(Q)&= \rp(f_a(Q_1))\lef a\rig \rp(f_a(Q_1)),
\end{align*}
and by induction, $Q_1=\rp(Q_1)=\rp(f_a(Q_1))$.
Hence, $\rp(Q)=Q_1\lef a\rig Q_1$, and
\begin{align*}
\rp(f_a(Q))&=\rp(f_a(f_a(Q_1)))\lef a\rig\rp(g_a(f_a(Q_1)))\\
&=\rp(f_a(Q_1))\lef a\rig\rp(f_a(Q_1))
&&\text{by Lemma~\ref{la:3.3}}\\
&=Q_1\lef a\rig Q_1,
\end{align*}
and the equalities for $R$ follow in a similar way.

If $Q= Q_1\lef a\rig Q_1$ and $R\ne R_1\lef a\rig R_1$,
the lemma's equalities follow in a similar way, although a bit simpler
because $g_a(R)=f_a(R)=R$.

For all remaining cases, the lemma's equalities  follow in a similar way.
\end{proof}

\begin{proposition}[{\normalfont $\rpbf$} is a normalization function]
\label{prop:3.7}
For all $P\in\PS$, $\rpbf(P)$ is an rp-basic form, and
for each rp-basic form $P$, $\rpbf(P)= P$.
\end{proposition}

\begin{proof}
The first statement is Lemma~\ref{la:3.5}. For the second statement, it suffices by
Lemma~\ref{la:2.14} to prove that
for each rp-basic form $P$, $\rp(P)=P$.
This follows by case distinction on $P$.
The cases $P\in\{\tr,\fa\}$ follow immediately, and otherwise $P= Q\lef a\rig R$, 
and thus  $\rp(P)=\rp(f_a(Q))\lef a\rig\rp(g_a(R))$. By Lemma~\ref{la:3.6}, 
$\rp(f_a(Q))=Q$ and $\rp(g_a(R))=R$, hence $\rp(P)=P$.
\end{proof}

\begin{lemma}
\label{la:3.8}
For all $P\in\BF$, $\CPrp\vdash P=\rp(P)$.
\end{lemma}

\begin{proof}
We apply structural induction on $P$.
The base cases $P\in\{\tr,\fa\}$ are trivial.
Assume $P= P_1\lef a\rig P_2$. By induction $\CPrp\vdash P_i=\rp(P_i)$. 
We proceed by a case distinction on the form 
that $P_1$ and $P_2$ can have: 
\begin{enumerate}\setlength\itemsep{-.1em}
\item 
If $P_i\in\{\tr,\fa,Q_i\lef b_i\rig Q_i'\}$ with $b_i\ne a$,
then $f_a(P_1)=P_1$ and $g_a(P_2)=P_2$, and hence
$\rp(P)= \rp(P_1)\lef a\rig \rp(P_2)$, and thus
$\CPrp\vdash P=\rp(P)$. 
\item 
If $P_1= R_1\lef a\rig R_2$ and $P_2\in\{\tr,\fa,Q'\lef b\rig Q''\}$ with $b\ne a$,
then $g_a(P_2)=P_2$ and by auxiliary result~\eqref{aux1} in the proof of Lemma~\ref{la:3.5}, 
$\rp(R_1)$ and $\rp(P_2)$ are rp-basic forms.
We derive
\begin{align*}
\CPrp&\vdash
P=(R_1\lef a\rig R_2)\lef a\rig P_2\\
&=(R_1\lef a\rig R_1)\lef a\rig P_2
&&\text{by~\eqref{CPrp1}}\\
&=(\rp(R_1)\lef a\rig \rp(R_1))\lef a\rig\rp(P_2)
&&\text{by IH}\\
&=(\rp(f_a(R_1))\lef a\rig \rp(f_a(R_1)))\lef a\rig\rp(g_a(P_2))
&&\text{by Lemma~\ref{la:3.6}}\\
&=\rp(f_a(R_1\lef a \rig R_2))\lef a\rig\rp(g_a(P_2))
\\
&=\rp((R_1\lef a \rig R_2)\lef a\rig P_2)
\\
&=\rp(P).
\end{align*}
\item 
If $P_1\in\{\tr,\fa,Q'\lef b\rig Q''\}$ with $b\ne a$ and 
$P_2= S_1\lef a\rig S_2$, we can proceed as in the previous case, but now
using axiom scheme~\eqref{CPrp2} and the identity $f_a(P_1)=P_1$, and the fact
that $\rp(P_1)$ and $\rp(S_2)$ are rp-basic forms.
\item 
If $P_1= R_1\lef a\rig R_2$ and $P_2= S_1\lef a\rig S_2$,
we can proceed as in two previous cases, now using both~\eqref{CPrp1} and \eqref{CPrp2},
and the fact that $\rp(R_1)$ and $\rp(S_2)$ are rp-basic forms.
\end{enumerate}
\end{proof}

\begin{theorem}
\label{thm:3.9}
For all $P\in\PS$, $\CPrp\vdash P=\rpbf(P)$.
\end{theorem}

\begin{proof}
By  Theorem~\ref{thm:1a} and Corollary~\ref{cor:1} we find
$\CPrp\vdash P=\baf(P)$. By Lemma~\ref{la:3.8},
$\CPrp\vdash \baf(P)=\rp(\baf(P))$, and $\rp(\baf(P))=\rpbf(P)$.
\end{proof}

\begin{definition}
\label{def:3.10}
The binary relation $=_\rpbf$ on \PS\
is defined as follows: 
\[P=_\rpbf Q~\iff~\rpbf(P)=\rpbf(Q).\]
\end{definition}

\begin{theorem}
\label{thm:3.11}
For all $P,Q\in\PS$, $\CPrp\vdash P=Q~\iff P=_\rpbf Q$.
\end{theorem}

\begin{proof}
($\Rightarrow$) Assume $\CPrp\vdash P=Q$. By Theorem~\ref{thm:3.9},
$\CPrp\vdash \rpbf(P)=\rpbf(Q)$. 
In~\cite{BP10} the following two
statements are proved (Theorem 6.3 and an auxiliary result in its proof),
where $=_\textit{rp}$ is a binary relation on $\PS$:
\begin{enumerate}\setlength\itemsep{-.1em}
\item
For all $P,Q\in\PS$,\quad
$\CPrp\vdash P=Q ~\iff~ P=_\textit{rp} Q$.
\item
For all rp-basic forms $P$ and $Q,\quad P=_\textit{rp} Q ~\Rightarrow~ P= Q$.
\end{enumerate}
By Lemma~\ref{la:3.5} 
these statements imply $\rpbf(P)=\rpbf(Q)$, that is, $P=_\rpbf Q$.

\smallskip

($\Leftarrow$)
Assume $P=_\rpbf Q$. By Theorem~\ref{thm:3.9}, $\CPrp\vdash P=Q$.
\end{proof}

So, the relation $=_\rpbf$ is axiomatized by $\CPrp$, and is thus a congruence.
With this observation in mind, we
define a transformation on evaluation trees that mimics the function
$\rpbf$, and prove that equality of two such transformed trees characterizes 
the congruence that is axiomatized by $\CPrp$.

\begin{definition}
\label{def:3.12}
The unary 
\textbf{repetition-proof evaluation function} 
\[\rpse:\PS\to \T\]
yields \textbf{repetition-proof evaluation trees} and is defined by
\begin{align*}
\rpse(P)&=\rpt(se(P)).
\end{align*}
The auxiliary function $\rpt:\T\to\T$ is defined as follows ($a\in A$):
\begin{align*}
\rpt(\tr)&=\tr,\\
\rpt(\fa)&=\fa,\\
\rpt(X\unlhd a\unrhd Y)&=\rpt(\f_a(X))\unlhd a\unrhd\rpt(\g_a(Y)).
\end{align*}
For $a\in A$, the auxiliary functions $\f_a: \T\to\T$ and $\g_a: \T\to\T$ 
are defined by
\begin{align*}
\f_a(\tr)&=\tr,\\ 
\f_a(\fa)&=\fa,\\
\f_a(X\unlhd b\unrhd Y)&=
\begin{cases}
\f_a(X)\unlhd a\unrhd \f_a(X)
&\text{if $b=a$},\\
X\unlhd b\unrhd Y
&\text{otherwise},
\end{cases}
\end{align*}
and
\begin{align*}
\g_a(\tr)&=\tr,\\
\g_a(\fa)&=\fa,\\
\g_a(X\unlhd b\unrhd Y)&=
\begin{cases}
\g_a(Y)\unlhd a\unrhd \g_a(Y)
&\text{if }b= a,\\
X\unlhd b\unrhd Y
&\text{otherwise}.
\end{cases}
\end{align*}
\end{definition}

\begin{example}\rm
\label{ex:rp}
Let $P=a\lef(\fa\lef a\rig\tr) \rig \fa$. 
We depict $se(P)$ (as in Example~\ref{ex:fr})
and the repetition-proof evaluation tree $\rpse(P)=\fa\unlhd a\unrhd(\fa\unlhd a\unrhd\fa)$:
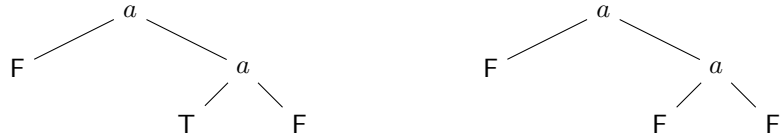
\[
\begin{array}[t]{ll}
\begin{array}[t]{l}
\\[-4mm]
\begin{tikzpicture}[%
      level distance=7.5mm,
      level 1/.style={sibling distance=30mm},
      level 2/.style={sibling distance=15mm},
      baseline=(current bounding box.center)]
      \node (a) {$a$}
        child {node (b1) {$\fa$}
        }
        child {node (b2) {$a$}
          child {node (d1) {$\tr$}} 
          child {node (d2) {$\fa$}}
        };
\end{tikzpicture}
\end{array}
&\qquad
\begin{array}[t]{l}
\\[-4mm]
\qquad
\begin{tikzpicture}[%
      level distance=7.5mm,
      level 1/.style={sibling distance=30mm},
      level 2/.style={sibling distance=15mm},
      baseline=(current bounding box.center)]
      \node (a) {$a$}
        child {node (b1) {$\fa$}
        }
        child {node (b2) {$a$}
          child {node (d1) {$\fa$}} 
          child {node (d2) {$\fa$}}
        };
\end{tikzpicture}
\\[-1mm]
\end{array}\end{array}
\]
\qedex\end{example} 
The similarities between $\rpse$ and the function $\rpbf$ can be exploited:
We use the following lemma in the proof of this section's last completeness result.

\begin{lemma}
\label{la:3.14}
For all $P\in\BF$, $\rpt(se(P))=se(\rp(P))$.
\end{lemma}

\begin{proof}
We first prove an auxiliary result:
\begin{equation}
\label{auxR4}
\text{For all $a\in A$ and $P\in\BF$, $\f_a(se(P))=se(f_a(P))$ and 
$\g_a(se(P))=se(g_a(P))$.}
\end{equation}
Fix some $a\in A$.
We prove~\eqref{auxR4} by structural induction on $P$.
The base cases $P\in\{\tr,\fa\}$ are trivial. 
For the inductive case $P= Q\lef b\rig R$ we have to distinguish
the cases $b=a$ and $b\ne a$. If $b=a$, then 
\begin{align*}
\f_a(se(Q\lef a\rig R))
&=\f_a(se(Q)\unlhd a\unrhd se(R))\\
&=\f_a(se(Q))\unlhd a\unrhd \f_a(se(Q))\\
&=se(f_a(Q))\unlhd a\unrhd se(f_a(Q))
&&\text{by IH}\\
&=se(f_a(Q\lef a\rig R)), 
\end{align*}
and if $b\ne a$, then
\begin{align*}
\f_a(se(Q\lef b\rig R))
&=\f_a(se(Q)\unlhd b\unrhd se(R))\\
&=se(Q)\unlhd b\unrhd se(R)\\
&=se(Q\lef b\rig R)\\
&=se(f_a(Q\lef b\rig R)).
\end{align*}
The second equality can be derived in a similar way, and
this finishes the proof of~\eqref{auxR4}.

\smallskip

We  prove the lemma's statement by induction on $d(P)$. 
The base cases $P\in\{\tr,\fa\}$
follow immediately. Assume $P= Q\lef a\rig R$, then
\begin{align*}
\rpt(se(Q\lef a\rig R))
&=\rpt(se(Q)\unlhd a\unrhd se(R))\\
&=\rpt(\f_a(se(Q)))\unlhd a\unrhd \rpt(\g_a(se(R)))\\
&=\rpt(se(f_a(Q)))\unlhd a\unrhd \rpt(se(g_a(R)))
&&\text{by \eqref{auxR4}}\\
&=se(\rp(f_a(Q)))\unlhd a\unrhd se(\rp(g_a(R)))
&&\text{by~IH (and Lemma~\ref{la:3.4})}\\
&=se(\rp(f_a(Q))\lef a\rig \rp(g_a(R)))\\
&=se(\rp(Q\lef a\rig R)).
\end{align*}
\end{proof}

Finally, we relate conditional statements by means of their  
repetition-proof evaluation trees.
\begin{definition}
\textbf{Repetition-proof valuation congruence}, notation $=_\rpse$, is defined 
on \PS\ as follows:
\[P=_\rpse Q~\iff~\rpse(P)=\rpse(Q).\]
\end{definition}

The following characterization result immediately implies that $=_\rpse$ is a 
congruence relation on $\PS$
(and hence justifies calling it a congruence).

\begin{proposition}
\label{prop:3.16}
For all $P,Q\in\PS$, $P=_\rpse Q~\iff~ P=_\rpbf Q$.
\end{proposition}

\begin{proof}
($\Rightarrow$) Assume $\rpse(P)=\rpse(Q)$, thus $\rpt(se(P))=\rpt(se(Q))$.
By Corollary~\ref{cor:1}, 
\\
$\rpt(se(\baf(P)))=\rpt(se(\baf(Q)))$,
so by Lemma~\ref{la:3.14},
$se(\rp(\baf(P)))=se(\rp(\baf(Q)))$.
\\
By Lemma~\ref{la:2.8} and auxiliary result~\eqref{aux1} 
(see the proof of Lemma~\ref{la:3.5}), it follows that $\rp(\baf(P))=\rp(\baf(Q))$,
that is, $P=_\rpbf Q$. 

\smallskip

($\Leftarrow$) Assume $P=_\rpbf Q$, thus $\rp(\baf(P))=\rp(\baf(Q))$ and
$se(\rp(\baf(P)))=se(\rp(\baf(Q)))$. By Lemma~\ref{la:3.14}, 
$\rpt(se(\baf(P)))=\rpt(se(\baf(Q)))$.
By Corollary~\ref{cor:1}, $se(\baf(P))=se(P)$ and $se(\baf(Q))=se(Q)$,
so $\rpt(se(P))=\rpt(se(Q))$, 
that is, $P=_\rpse Q$.
\end{proof}

We end this section with the completeness result we were after.

\begin{theorem}[Completeness of $\CPrp$]
\label{thm:3.17}
For all $P,Q\in\PS$, 
\[\CPrp\vdash P=Q
~\iff~ 
P=_\rpse Q.\]
\end{theorem}

\begin{proof}
Combine Theorem~\ref{thm:3.11} and Proposition~\ref{prop:3.16}. 
\end{proof}

\section{Evaluation trees for contractive valuation congruence}
\label{sec:cr}
In~\cite{BP10} we introduced $\CPcon$, 
\emph{contractive} \CP, as the extension of \CP\ with the following two axiom schemes,
where $a$ ranges over $A$: 
\begin{align*}
\label{CPcr1}
\tag{CPcr1}
(x\lef a\rig y)\lef a\rig z&=x\lef a\rig z,\\
\label{CPcr2}
\tag{CPcr2}
x\lef a\rig (y\lef a\rig z)&=x\lef a\rig z.
\end{align*}
These schemes prescribe contraction for each atom $a$ for
respectively the \emph{true}-case and the 
\emph{false}-case, and are each others dual.
It easily follows that the axiom schemes~\eqref{CPrp1}
and~\eqref{CPrp2} are derivable from $\CPcon$,
so $\CPcon$ is also an axiomatic extension
of $\CPrp$. 

Again, we define a proper subset of basic forms with the property that each
propositional statement can be proved equal to such a basic form.

\begin{definition}
\label{def:4.1}
\textbf{Cr-basic forms} are inductively defined:
\begin{itemize}\setlength\itemsep{-.1em}
\item
 \tr\ and \fa\ are cr-basic forms, and
\item 
 $P_1\lef a\rig P_2$ is a cr-basic form if $P_1$
 and $P_2$ are cr-basic forms, and if $P_i$ is not equal to \tr\ or \fa,
 the central condition in $P_i$ is different from $a$.
\end{itemize}
\end{definition}

It will turn out useful to define a function that transforms conditional statements
into cr-basic forms,
and that is comparable to the function $\baf$ (see Definition~\ref{def:bf}).

\begin{definition}
\label{def:4.2}
The \textbf{cr-basic form function}
$\crbf:\PS\to \PS$
is defined by
\begin{align*}
\crbf(P)&=\con(\baf(P)).
\end{align*}
The auxiliary function $\con:\BF\to\BF$ is defined as follows:
\begin{align*}
\con(\tr)&=\tr\\
\con(\fa)&=\fa,\\
\con(P\lef a\rig Q)&=\con(\ci_a(P))\lef a\rig\con(\cj_a(Q)).
\end{align*}
For $a\in A$, the auxiliary functions
$\ci_a: \BF\to\BF$ and $\cj_a: \BF\to\BF$ 
are defined by
\begin{align*}
\ci_a(\tr)&=\tr,\\
\ci_a(\fa)&=\fa,\\
\ci_a(P\lef b\rig Q)&=\begin{cases}
\ci_a(P)&\text{if } b=a,\\
P\lef b\rig Q&\text{otherwise},
\end{cases}
\end{align*}
and
\begin{align*}
\cj_a(\tr)&=\tr,\\
\cj_a(\fa)&=\fa,\\
\cj_a(P\lef b\rig Q)&=\begin{cases}
\cj_a(Q)&\text{if } b=a,\\
P\lef b\rig Q&\text{otherwise}.
\end{cases}
\end{align*}
\end{definition}

Thus, $\crbf$ maps a conditional statement $P$ to $\baf(P)$ and then
transforms $\baf(P)$ according to the auxiliary functions $\con$, $\ci_a$, and $\cj_a$.

\begin{lemma}
\label{la:4.3}
For all $a\in A$ and $P\in\BF$, $d(P)\geq d(\ci_a(P))$ and $d(P)\geq d(\cj_a(P))$.
\end{lemma}

\begin{proof}
Fix some $a\in A$.
We  prove these inequalities by structural induction on $P$.  
The base cases $P\in\{\tr,\fa\}$
are trivial. For the inductive case $P= Q\lef b\rig R$ 
we have to distinguish the cases $b=a$ and $b\ne a$. 
If $b=a$, then 
\begin{align*}
d(Q\lef a\rig R)
&=1+\max\{d(Q),d(R)\}\\
&\geq 1+ d(Q)\\
&\geq 1+ d(\ci_a(Q))
&&\text{by IH}\\
&=1+d(\ci_a(Q\lef a\rig R)),
\end{align*}
and $d(Q\lef a\rig R)\geq d(\cj_a(Q\lef a\rig R))$ follows in a similar way.

If $b\ne a$, then
$\ci_a(P)=\cj_a(P)=P$, and hence $d(P)\geq d(\ci_a(P))$
and $d(P)\geq d(\cj_a(P))$.
\end{proof}

\begin{lemma}
\label{la:4.4}
For all $P\in\PS$, $\crbf(P)$ is a cr-basic form.
\end{lemma}

\begin{proof}
We first prove an auxiliary result:
\begin{equation}
\label{aux2}
\text{For all $P\in\BF$, $\con(P)$ is a cr-basic form.}
\end{equation}
This follows by induction on the depth $d(P)$ of $P$. 
If $d(P)=0$, then $P\in\{\tr,\fa\}$, and hence $\con(P)=P$ is a cr-basic form.
For the inductive case $d(P)=n+1$ it must be the case that $P= Q\lef a\rig R$.
We find
\begin{align*}
\con(Q\lef a\rig R)
&=\con(\ci_a(Q))\lef a\rig \con(\cj_a(R)),
\end{align*}
which is a cr-basic form because
\begin{itemize}\setlength\itemsep{-.1em}
\item[$-$] 
 by Lemma~\ref{la:4.3}, $\ci_a(Q)$ and $\cj_a(R)$ are basic forms with depth smaller 
than or equal to $n$,
so by the induction hypothesis, $\con(\ci_a(Q))$ and $\con(\cj_a(R))$ are cr-basic forms,
\item[$-$] 
by definition of the auxiliary functions $\ci_a$ and $\cj_a$, the central condition of
$\ci_a(Q)$ and $\cj_a(R)$ is not equal to $a$, hence 
$\con(\ci_a(Q))\lef a\rig \con(\cj_a(R))$ is a cr-basic form.
\end{itemize}
This completes the proof of~\eqref{aux2}. 

\smallskip

The lemma's statement now follows by structural induction: 
the base cases (comprising a single atom $a$) are again trivial, 
and for the inductive case, 
\[\crbf(P\lef Q\rig R)=\con(\baf(P\lef Q\rig R))= \con(S)\]
for some basic form $S$
by Lemma~\ref{la:2.13}, and by~\eqref{aux2}, $\con(S)$
is a cr-basic form.
\end{proof}

The following, somewhat technical lemma is used in Proposition~\ref{prop:4.6} and 
Lemma~\ref{la:4.7}. 

\begin{lemma}
\label{la:4.5}
If $Q\lef a\rig R$ is a cr-basic form, then $Q=\con(Q)=\con(\ci_a(Q))$ and
$R=\con(R)=\con(\cj_a(R))$.
\end{lemma}

\begin{proof}
By simultaneous induction
on the structure of $Q$ and $R$. 
The base case, thus $Q,R\in\{\tr,\fa\}$,
is again trivial.
If $Q= Q_1\lef b\rig Q_2$ and $R= R_1\lef c\rig R_2$, then $b\ne a\ne c$ and
thus $\ci_a(Q)=Q$ and $\cj_a(R)=R$.
Moreover, $Q_1$ and $Q_2$ have no central condition $b$, hence $\ci_b(Q_1)=Q_1$ 
and $\cj_b(Q_2)=Q_2$, and thus
\begin{align*}
\con(Q)&=\con(\ci_b(Q_1))\lef b\rig \con(\cj_b(Q_2))\\
&=\con(Q_1)\lef b\rig \con(Q_2)\\
&=Q_1\lef b\rig Q_2.
&&\text{by IH}
\end{align*} 
The equalities for $R$ follow in a similar way.

\smallskip

If $Q= Q_1\lef b\rig Q_1$ and $R\in\{\tr,\fa\}$,
the lemma's equalities follow in a similar way, and this is also the
case if $Q\in\{\tr,\fa\}$ and $R= Q_1\lef b\rig Q_1$.
\end{proof}

With Lemma~\ref{la:4.5} we can easily prove the following result.

\begin{proposition}[{\normalfont $\crbf$} is a normalization function]
\label{prop:4.6}
For each $P\in\PS$, $\crbf(P)$ is a cr-basic form, and
for each cr-basic form $P$, $\crbf(P)= P$.
\end{proposition}

\begin{proof}
The first statement is Lemma~\ref{la:4.4}. For the second statement, it suffices by
Lemma~\ref{la:2.14} to prove that
$\con(P)=P$. We prove this by case distinction on $P$.
The cases $P\in\{\tr,\fa\}$ follow immediately, and otherwise $P= Q\lef a\rig R$, 
and thus  $\con(P)=\con(\ci_a(Q))\lef a\rig\con(\cj_a(R))$. By Lemma~\ref{la:4.5}, 
$\con(\ci_a(Q))=Q$ and $\con(\cj_a(R))=R$, hence $\con(P)=P$.
\end{proof}

\begin{lemma}
\label{la:4.7}
For all $P\in\BF$, $\CPcr\vdash P=\con(P)$.
\end{lemma}

\begin{proof}
We first prove two auxiliary results:
\begin{align}
\label{eq:C1}
\text{For all $a\in A$ and $P,Q\in\BF$,~}
&\CPcr\vdash P\lef a\rig Q=P\lef a\rig \cj_a(Q),\\
\label{eq:C2}
&\CPcr\vdash P\lef a\rig Q=\ci_a(P)\lef a\rig Q.
\end{align}
Fix some $a\in A$.
We prove~\eqref{eq:C1} by structural induction on $Q$. 
The base cases $Q\in\{\tr,\fa\}$ are trivial.
For the inductive case $Q= Q_1\lef b\rig Q_2$ we have to distinguish
the cases $b=a$ and $b\ne a$. If $b=a$, then
$\cj_a(Q)=\cj_a(Q_2)$ and
\begin{align*}
\CPcr\vdash
P\lef a\rig(Q_1\lef a\rig Q_2)
&=P\lef a\rig Q_2
&&\text{by \eqref{CPcr2}}\\
&=P\lef a\rig \cj_a(Q_2)
&&\text{by IH}\\
&=P\lef a\rig \cj_a(Q).
\end{align*}
If $b\ne a$ then 
$\cj_a(Q)=Q$, hence
$\CPcr\vdash
P\lef a\rig Q=P\lef a\rig \cj_a(Q)$.

\smallskip

Auxiliary result~\eqref{eq:C2} follows in a similar way by structural
induction on $P$ and with help of axiom scheme \eqref{CPcr1}.

\medskip

The lemma's statement follows by induction on $d(P)$.
The base cases $P\in\{\tr,\fa\}$ are trivial.
For the inductive case, assume $P= Q\lef a\rig R$. We derive
\begin{align*}
\CPcr\vdash
Q\lef a\rig R
&=\ci_a(Q)\lef a\rig \cj_a(R)
&&\text{by \eqref{eq:C1}, \eqref{eq:C2}}\\
&=\con(\ci_a(Q))\lef a\rig \con(\cj_a(R))
&&\text{by~IH (and Lemma~\ref{la:4.3})}\\
&=\con(Q\lef a\rig R).
\end{align*}
\end{proof}

\begin{theorem}
\label{thm:4.8}
For all $P\in\PS$, $\CPcr\vdash P=\crbf(P)$.
\end{theorem}

\begin{proof}
By Theorem~\ref{thm:1a} and Corollary~\ref{cor:1}, $\CPcr\vdash P=\baf(P)$, 
and by Lemma~\ref{la:4.7},
\\
$\CPcr\vdash \baf(P)=\con(\baf(P))$, and $\con(\baf(P))=\crbf(P)$.
\end{proof}

\begin{definition}
\label{def:4.9}
The binary relation $=_\crbf$ on \PS\
is defined as follows: 
\[P=_\crbf Q~\iff~\crbf(P)=\crbf(Q).\]
\end{definition}

\newpage

\begin{theorem}
\label{thm:4.10}
For all $P,Q\in\PS$, $\CPcr\vdash P=Q~\iff P=_\crbf Q$.
\end{theorem}

\begin{proof}
($\Rightarrow$) Assume $\CPcr\vdash P=Q$. Then, by Theorem~\ref{thm:4.8},
$\CPcr\vdash \crbf(P)=\crbf(Q)$. 
In~\cite{BP10} the following two
statements are proved (Theorem 6.4 and an auxiliary result in its proof),
where $=_\textit{cr}$ is a binary relation on $\PS$:
\begin{enumerate}\setlength\itemsep{-.1em}
\item
For all $P,Q\in\PS$,\quad
$\CPcr\vdash P=Q ~\iff~ P=_\textit{cr} Q$.
\item
For all cr-basic forms $P$ and $Q,\quad P=_\textit{cr} Q ~\Rightarrow~ P= Q$.
\end{enumerate}
By Lemma~\ref{la:4.4},
these statements imply $\crbf(P)=\crbf(Q)$, that is, $P=_\crbf Q$.

\smallskip

($\Leftarrow$) Assume $P=_\crbf Q$. By Theorem~\ref{thm:4.8}, $\CPcr\vdash P=Q$.
\end{proof}

Hence, the relation $=_\crbf$ is axiomatized by $\CPcr$, and is thus a congruence.
We now define a transformation on evaluation trees that mimics the function
$\crbf$, and prove that equality of two such transformed trees characterizes 
the congruence that is axiomatized by $\CPcr$.

\begin{definition}
\label{def:4.11}
The unary 
\textbf{contractive evaluation function} 
\[\crse:\PS\to \T\]
yields \textbf{contractive evaluation trees} and is defined by
\begin{align*}
\crse(P)&=\crt(se(P)).
\end{align*}
The auxiliary function $\crt:\T\to\T$ is defined as follows ($a\in A$):
\begin{align*}
\crt(\tr)&=\tr,\\
\crt(\fa)&=\fa,\\
\crt(X\unlhd a\unrhd Y)&=\crt(\Fi_a(X))\unlhd a\unrhd\crt(\Fj_a(Y)).
\end{align*}
For $a\in A$, the auxiliary functions $\Fi_a: \T\to\T$ and $\Fj_a: \T\to\T$ 
are defined by
\begin{align*}
\Fi_a(\tr)&=\tr,\\
\Fi_a(\fa)&=\fa,\\
\Fi_a(X\unlhd b\unrhd Y)&=
\begin{cases}
\Fi_a(X)
&\text{if $b=a$},\\
X\unlhd b\unrhd Y
&\text{otherwise},
\end{cases}
\end{align*}
and
\begin{align*}
\Fj_a(\tr)&=\tr,\\
\Fj_a(\fa)&=\fa,\\
\Fj_a(X\unlhd b\unrhd Y)&=
\begin{cases}
\Fj_a(Y)
&\text{if }b= a,\\
X\unlhd b\unrhd Y
&\text{otherwise}.
\end{cases}
\end{align*}
\end{definition}

As a simple example we depict $se((a\lef a\rig\fa)\lef a\rig\fa)$
and the contractive evaluation tree $\crse((a\lef a\rig\fa)\lef a\rig\fa)$:
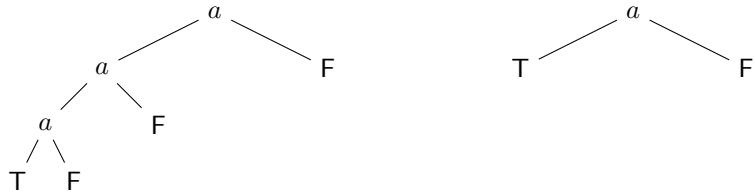
\[
\begin{array}[t]{ll}
\begin{array}[t]{l}
\\[-4mm]
\begin{tikzpicture}[%
level distance=7.5mm,
level 1/.style={sibling distance=30mm},
level 2/.style={sibling distance=15mm},
level 3/.style={sibling distance=7.5mm}
]
\node (a) {$a$}
  child {node (b1) {$a$}
    child {node (c1) {$a$}
      child {node (d1) {$\tr$}} 
      child {node (d2) {$\fa$}}
    }
    child {node (c2) {$\fa$}
    }
  }
  child {node (b2) {$\fa$} 
  };
\end{tikzpicture}
\end{array}
&\qquad
\begin{array}[t]{l}
\\[-4mm]
\qquad
\begin{tikzpicture}[%
level distance=7.5mm,
level 1/.style={sibling distance=30mm},
level 2/.style={sibling distance=15mm},
level 3/.style={sibling distance=7.5mm}
]
\node (a) {$a$}
    child {node (c1) {$\tr$}
  }
  child {node (b2) {$\fa$}
  };
\end{tikzpicture}
\\[8mm]
\end{array}\end{array}
\]

The similarities between the evaluation function
$\crse$ and the function $\crbf$ can be exploited,
and we use the following lemma in the proof of the next completeness result.

\begin{lemma}
\label{la:4.12}
For all $P\in\BF$, $\crt(se(P))=se(\con(P))$.
\end{lemma}

\begin{proof}
We first prove the following auxiliary result:
\begin{equation}
\label{auxC4}
\text{For all $a\in A$ and $P\in\BF$, $\Fi_a(se(P))=se(\ci_a(P))$ and 
$\Fj_a(se(P))=se(\cj_a(P))$.}
\end{equation}
Fix some $a\in A$.
We prove~\eqref{auxC4} by structural induction on $P$.
The base cases $P\in\{\tr,\fa\}$ are trivial. 
For the inductive case $P= Q\lef b\rig R$ we have to distinguish
the cases $b=a$ and $b\ne a$. If $b=a$, then 
\begin{align*}
\Fi_a(se(Q\lef a\rig R))
&=\Fi_a(se(Q)\unlhd a\unrhd se(R))\\
&=\Fi_a(se(Q))\\
&=se(\ci_a(Q))
&&\text{by IH}\\
&=se(\ci_a(Q\lef a\rig R)), 
\end{align*}
and if $b\ne a$, then
\begin{align*}
\Fi_a(se(Q\lef b\rig R))
&=\Fi_a(se(Q)\unlhd b\unrhd se(R))\\
&=se(Q)\unlhd b\unrhd se(R)\\
&=se(Q\lef b\rig R)\\
&=se(\ci_a(Q\lef b\rig R)).
\end{align*}
The second equality can be derived in a similar way, and
this finishes the proof of~\eqref{auxC4}.

\smallskip

We  prove the lemma's statement by induction on $d(P)$. 
The base cases $P\in\{\tr,\fa\}$
follow immediately. Assume $P= Q\lef a\rig R$, then
\begin{align*}
\crt(se(Q\lef a\rig R))
&=\crt(se(Q)\unlhd a\unrhd se(R))\\
&=\crt(\Fi_a(se(Q)))\unlhd a\unrhd \crt(\Fj_a(se(R)))\\
&=\crt(se(\ci_a(Q)))\unlhd a\unrhd \crt(se(\cj_a(R)))
&&\text{by \eqref{auxC4}}\\
&=se(\con(\ci_a(Q)))\unlhd a\unrhd se(\con(\cj_a(R)))
&&\text{by~IH (and Lemma~\ref{la:4.3})}\\
&=se(\con(\ci_a(Q))\lef a\rig \con(\cj_a(R)))\\
&=se(\con(Q\lef a\rig R)).
\end{align*}
\end{proof}

Finally, we relate conditional statements by means of their contractive evaluation trees.
\begin{definition}
\label{def:4.13}
\textbf{Contractive valuation congruence}, notation $=_\crse$, is defined on \PS\
as follows:
\[P=_\crse Q~\iff~\crse(P)=\crse(Q).\]
\end{definition}

The following characterization result immediately implies that $=_\crse$ is a congruence 
relation on $\PS$ (and hence justifies calling it a congruence).

\begin{proposition}
\label{prop:4.14}
For all $P,Q\in\PS$, $P=_\crse Q~\iff~ P=_\crbf Q$.
\end{proposition}

\begin{proof}
($\Rightarrow$) Assume $\crse(P)=\crse(Q)$, thus $\crt(se(P))=\crt(se(Q))$.
By Corollary~\ref{cor:1}, 
\\
$\crt(se(\baf(P)))=\crt(se(\baf(Q)))$, so by Lemma~\ref{la:4.12},
$se(\con(\baf(P)))=se(\con(\baf(Q)))$.
\\
By Lemma~\ref{la:2.8} and auxiliary result~\eqref{aux1} 
(see the proof of Lemma~\ref{la:3.5}), it follows that $\con(\baf(P))=\con(\baf(Q))$,
that is, $P=_\crbf Q$. 

\smallskip

($\Leftarrow$) Assume $P=_\crbf Q$, thus $\con(\baf(P))=\con(\baf(Q))$ and
$se(\con(\baf(P)))=se(\con(\baf(Q)))$. By Lemma~\ref{la:4.12}, 
$\crt(se(\baf(P)))=\crt(se(\baf(Q)))$.
By Corollary~\ref{cor:1},
$se(\baf(P))=se(P)$ and $se(\baf(Q))=se(Q)$,
so $\crt(se(P))=\crt(se(Q))$, that is, $P=_\crse Q$.
\end{proof}

Our final result in this section is a completeness result for 
contractive valuation congruence.

\begin{theorem}[Completeness of $\CPcr$]
\label{thm:4.15}
For all $P,Q\in\PS$, 
\[\CPcr\vdash P=Q
~\iff~ 
P=_\crse Q.\]
\end{theorem}

\begin{proof}
Combine Theorem~\ref{thm:4.10} and Proposition~\ref{prop:4.14}.
\end{proof}

\section{Evaluation trees for memorizing valuation congruence}
\label{sec:mem}
In~\cite{BP10} we introduced $\CPmem$, \emph{memorizing \CP}, as the 
extension of \CP\ with the following axiom:
\begin{align*}
\label{CPmem}
\tag{CPmem} 
\qquad
x\lef y\rig(z\lef u\rig(v\lef y\rig w))
&= x\lef y\rig(z\lef u\rig w).
\end{align*}
Axiom~\eqref{CPmem} 
expresses that the first evaluation value of $y$
is memorized. More precisely, a 
``memorizing evaluation" is one with the property that
upon the evaluation of a compound propositional statement,
the first evaluation value of each atom is memorized throughout the 
evaluation.
We write $\CPmem$
for the set $\CP\cup\{\eqref{CPmem}\}$ of axioms. 

Replacing the variable $y$ in axiom~\eqref{CPmem} by 
$\fa\lef y\rig\tr$ and/or the variable $u$ by $\fa\lef u\rig\tr$
yields all other memorizing patterns:
\begin{align}
\label{CPmem'}
\tag{CPm1}
\qquad
(z\lef u\rig(w\lef y\rig v))\lef y\rig x&=
(z\lef u\rig w)\lef y\rig x,\\
\label{CPmem''}
\tag{CPm2}
x\lef y\rig((v\lef y\rig w)\lef u\rig z)
&= x\lef y\rig(w\lef u\rig z),\\
\label{CPmem'''}
\tag{CPm3}
((w\lef y\rig v)\lef u\rig z)\lef y\rig x&=
(w\lef u\rig z)\lef y\rig x.
\end{align}
Hence, the duality priciple also holds in $\CPmem$.
Furthermore, if we replace in axiom~\eqref{CPmem} $u$ by $\fa$,
we find the \emph{contraction law}
\label{p:contr}
\begin{equation}
\label{eq:contr}
\qquad
x\lef y\rig(v\lef y\rig w)=x\lef y\rig w,
\end{equation} 
and replacing $y$ by $\fa\lef y\rig\tr$ 
then yields the dual contraction law 
\begin{equation}
\label{eq:contr2}
\qquad
(w\lef y\rig v)\lef y\rig x= w\lef y\rig x.
\end{equation}
Hence, $\CPmem$ is an axiomatic extension of $\CPcr$.

We define a proper subset of basic forms with the property that each
propositional statement can be proved equal to such a basic form. 

\begin{definition}
\label{def:5.1}
Let $A'$ be a subset of A. \textbf{Mem-basic
forms over $A'$} are inductively defined:
\begin{itemize}\setlength\itemsep{-.1em}
\item
 \tr\ and \fa\ are mem-basic forms over $A'$, and
\item
 $P\lef a\rig Q$ is a mem-basic form over $A'$ if $a\in A'$ and $P$
 and $Q$ are mem-basic forms over $A'\setminus \{a\}$.
\end{itemize}
$P$ is a \textbf{mem-basic form} if for some $A'\subset A$, $P$ is a mem-basic form
over $A'$.
\end{definition}

Note that if $A$ is finite, the number of mem-basic forms is also finite.
It will turn out useful to define a function that transforms conditional statements
into mem-basic forms.

\begin{definition}
\label{def:5.2}
The \textbf{mem-basic form function}
$\membf:\PS\to \PS$
is defined by
\begin{align*}
\membf(P)&=\mem(\baf(P)).
\end{align*}
The auxiliary function $\mem:\BF\to\BF$ is defined as follows:
\begin{align*}
\mem(\tr)&=\tr\\
\mem(\fa)&=\fa,\\
\mem(P\lef a\rig Q)&=\mem(\ell_a(P))\lef a\rig\mem(\ri_a(Q)).
\end{align*}
For $a\in A$, the auxiliary functions
$\ell_a: \BF\to\BF$ and $\ri_a: \BF\to\BF$ 
are defined by
\begin{align*}
\ell_a(\tr)&=\tr,\\
\ell_a(\fa)&=\fa,\\
\ell_a(P\lef b\rig Q)&=\begin{cases}
\ell_a(P)&\text{if } b=a,\\
\ell_a(P)\lef b\rig \ell_a(Q)&\text{otherwise},
\end{cases}
\end{align*}
and
\begin{align*}
\ri_a(\tr)&=\tr,\\
\ri_a(\fa)&=\fa,\\
\ri_a(P\lef b\rig Q)&=\begin{cases}
\ri_a(Q)&\text{if } b=a,\\
\ri_a(P)\lef b\rig \ri_a(Q)&\text{otherwise}.
\end{cases}
\end{align*}
\end{definition}

Thus, $\membf$ maps a conditional statement $P$ to $\baf(P)$ and then
transforms $\baf(P)$ according to the auxiliary functions $\mem$,
$\ell_a$, and $\ri_a$.
We will use the following inequalities.

\begin{lemma}
\label{la:5.3}
For all $a\in A$ and $P\in\BF$, $d(P)\geq d(\ell_a(P))$ and $d(P)\geq d(\ri_a(P))$.
\end{lemma}

\begin{proof}
Fix some $a\in A$.
We  prove these inequalities by structural induction on $P$.  
The base cases $P\in\{\tr,\fa\}$
are trivial. For the inductive case $P= Q\lef b\rig R$ 
we have to distinguish the cases $b=a$ and $b\ne a$. 
If $b=a$, then 
\begin{align*}
d(Q\lef a\rig R)
&=1+\max\{d(Q),d(R)\}\\
&\geq 1+ d(Q)\\
&\geq 1+ d(\ell_a(Q))
&&\text{by IH}\\
&=1+d(\ell_a(Q\lef a\rig R)),
\end{align*}
and $d(Q\lef a\rig R)\geq d(\ri_a(Q\lef a\rig R))$ follows in a similar way.

If $b\ne a$, then
\begin{align*}
d(Q\lef b\rig R)
&=1+\max\{d(Q),d(R)\}\\
&\geq 1+ \max\{d(\ell_a(Q)),d(\ell_a(R))\}
&&\text{by IH}\\
&=d(\ell_a(Q)\lef b\rig \ell_a(R))\\
&=d(\ell_a(Q\lef b\rig R)),
\end{align*}
and $d(Q\lef b\rig R)\geq d(\ri_a(Q\lef b\rig R))$ follows in a similar way.
\end{proof}

\begin{lemma}
\label{la:5.4}
For all $P\in\PS$, $\membf(P)$ is a mem-basic form.
\end{lemma}

\begin{proof}
We first prove an auxiliary result:
\begin{equation}
\label{aux3}
\text{For all $P\in\BF$, $\mem(P)$ is a mem-basic form.}
\end{equation}
This follows by induction on the depth $d(P)$ of $P$. 
If $d(P)=0$, then $P\in\{\tr,\fa\}$, and hence $\mem(P)=P$ is a mem-basic form.
For the inductive case $d(P)=n+1$ it must be the case that $P= Q\lef a\rig R$.
We find
\begin{align*}
\mem(Q\lef a\rig R)
&=\mem(\ell_a(Q))\lef a\rig \mem(\ri_a(R)),
\end{align*}
which is a mem-basic form because
by Lemma~\ref{la:5.3}, $\ell_a(Q)$ and $\ri_a(R)$ are basic forms with depth smaller 
than or equal to $n$,
so by the induction hypothesis, $\mem(\ell_a(Q))$ is a mem-basic form over $A_Q$ 
and $\mem(\ri_a(R))$ is a mem-basic form over $A_R$ for suitable subsets $A_Q$ and $A_R$ of $A$.
Notice that by definition of $\ell_a$ and $\ri_a$ we can assume that 
the atom $a$ does not occur in $A_Q\cup A_R$.
Hence, $\mem(\ell_a(Q))\lef a\rig \mem(\ri_a(R))$ is a mem-basic form over $A_Q\cup A_R\cup\{a\}$,
which completes the proof of~\eqref{aux3}. 

\smallskip

The lemma's statement now follows by structural induction: 
the base cases (comprising a single atom $a$) are again trivial, 
and for the inductive case, 
\[\membf(P\lef Q\rig R)=\mem(\baf(P\lef Q\rig R))= \mem(S)\]
for some basic form $S$
by Lemma~\ref{la:2.13}, and by~\eqref{aux3}, $\mem(S)$
is a mem-basic form.
\end{proof}

With Lemma~\ref{la:5.4} we can easily prove the following result.

\begin{proposition}[{\normalfont $\membf$} is a normalization function]
\label{prop:5.5}
For each $P\in\PS$, $\membf(P)$ is a mem-basic form, and
for each mem-basic form $P$, $\membf(P)= P$.
\end{proposition}

\begin{proof}
The first statement is Lemma~\ref{la:5.4}. 
For the second statement, it suffices by
Lemma~\ref{la:2.14} to prove that
$\mem(P)=P$. We prove this by induction on $d(P)$.
The base cases $P\in\{\tr,\fa\}$ are trivial, and 
for the inductive case, assume $P= Q\lef a\rig R$, 
thus  $\mem(P)=\mem(\ell_a(Q))\lef a\rig\mem(\ri_a(R))$. 
Because $P$ is a mem-basic form, $Q$
and $R$ are mem-basic forms in which $a$ does not occur, and thus
$\ell_a(Q)=Q$ and $\ri_a(R)=R$. By induction, $\mem(Q)=Q$ and
$\mem(R)=R$, and thus $\mem(P)=P$.
\end{proof}

\begin{lemma}
\label{la:5.6}
For all $P\in\BF$, $\CPmem\vdash P=\mem(P)$.
\end{lemma}

\begin{proof}
We first prove an auxiliary result:
\begin{align}
\label{eq:11}
\text{For all $a\in A$ and $P,Q\in\BF,~\CPmem\vdash~$}
&P\lef a\rig Q=P\lef a\rig \ri_a(Q),\\
\label{eq:12}
&P\lef a\rig Q=\ell_a(P)\lef a\rig Q.
\end{align}
Fix some $a\in A$.
We prove~\eqref{eq:11} by structural induction on $Q$. 
The base cases $Q\in\{\tr,\fa\}$ are trivial.
For the inductive case $Q= Q_1\lef b\rig Q_2$ we have to distinguish
the cases $b=a$ and $b\ne a$. If $b=a$, then
$\ri_a(Q)=\ri_a(Q_2)$ and
\begin{align*}
\CPmem\vdash
P\lef a\rig(Q_1\lef a\rig Q_2)
&=P\lef a\rig Q_2
&&\text{by \eqref{eq:contr}}\\
&=P\lef a\rig \ri_a(Q_2)
&&\text{by IH}\\
&=P\lef a\rig \ri_a(Q).
\end{align*}
If $b\ne a$, then 
$\ri_a(Q)=\ri_a(Q_1)\lef b\rig\ri_a(Q_2)$ and
\begin{align*}
\CPmem&\vdash
P\lef a\rig(Q_1\lef b\rig Q_2)\\
&=P\lef a\rig ((\tr\lef a\rig Q_1)\lef b\rig(\tr\lef a\rig Q_2))
&&\text{by \eqref{CPmem}, \eqref{CPmem''}}\\
&=P\lef a\rig ((\tr\lef a\rig \ri_a(Q_1))\lef b\rig(\tr\lef a\rig \ri_a(Q_2)))
&&\text{by IH (twice)}\\
&=P\lef a\rig (\ri_a(Q_1)\lef b\rig\ri_a(Q_2))
&&\text{by \eqref{CPmem}, \eqref{CPmem''}}\\
&=P\lef a\rig \ri_a(Q).
\end{align*}
Auxiliary result~\eqref{eq:12} follows in a similar way with help of
axioms~\eqref{CPmem'} and~\eqref{CPmem'''}.

\smallskip

The lemma's statement follows by induction on $d(P)$.
The base cases $P\in\{\tr,\fa\}$ are trivial.
For the inductive case, assume $P= Q\lef a\rig R$. We derive
\begin{align*}
\CPmem\vdash
Q\lef a\rig R
&=\ell_a(Q)\lef a\rig \ri_a(R)
&&\text{by \eqref{eq:11}, \eqref{eq:12}}\\
&=\mem(\ell_a(Q))\lef a\rig \mem(\ri_a(R))
&&\text{by~IH (and Lemma~\ref{la:5.3})}\\
&=\mem(Q\lef a\rig R).
\end{align*}
\end{proof}

\begin{theorem}
\label{thm:5.7}
For all $P\in\PS$, $\CPmem\vdash P=\membf(P)$.
\end{theorem}

\begin{proof}
By Theorem~\ref{thm:1a} and Corollary~\ref{cor:1}, $\CPmem\vdash P=\baf(P)$, 
and by Lemma~\ref{la:5.6},
$\CPmem\vdash \baf(P)=\mem(\baf(P))$, and $\mem(\baf(P))=\membf(P)$.
\end{proof}

\begin{definition}
\label{def:5.8}
The binary relation $=_\membf$ on \PS\
is defined as follows: 
\[P=_\membf Q~\iff~\membf(P)=\membf(Q).\]
\end{definition}

\begin{theorem}
\label{thm:5.9}
For all $P,Q\in\PS$, $\CPmem\vdash P=Q~\iff P=_\membf Q$.
\end{theorem}

\begin{proof}
($\Rightarrow$) Assume $\CPmem\vdash P=Q$. Then, by Theorem~\ref{thm:5.7},
$\CPmem\vdash \membf(P)=\membf(Q)$. 
In~\cite{BP10} the following two
statements are proved (Theorem~8.1 and Lemma~8.4),
where $=_\textit{mem}$ is a binary relation on $\PS$:
\begin{enumerate}\setlength\itemsep{-.1em}
\item
For all $P,Q\in\PS$,\quad
$\CPmem\vdash P=Q ~\iff~ P=_\textit{mem} Q$.
\item
For all mem-basic forms $P$ and $Q,\quad P=_\textit{mem} Q ~\Rightarrow~ P= Q$.
\end{enumerate}
By Lemma~\ref{la:5.4} 
these statements imply $\membf(P)=\membf(Q)$, that is, $P=_\membf Q$.

\smallskip

($\Leftarrow$) Assume $P=_\membf Q$. By Theorem~\ref{thm:5.7},
$\CPmem\vdash P=Q$.
\end{proof}

Hence, the relation $=_\membf$ is axiomatized by $\CPmem$ and is thus a congruence. 
We define a transformation on evaluation trees that mimics the function
$\membf$, and prove that equality of two such transformed trees characterizes 
the congruence that is axiomatized by $\CPmem$.

\begin{definition}
\label{def:5.10}
The unary 
\textbf{memorizing evaluation function} 
\[\memse:\PS\to \T\]
yields \textbf{memorizing evaluation trees} and is defined by
\begin{align*}
\memse(P)&=\memt(se(P)).
\end{align*}
The auxiliary function $\memt:\T\to\T$ is defined as follows ($a\in A$):
\begin{align*}
\memt(\tr)&=\tr,\\
\memt(\fa)&=\fa,\\
\memt(X\unlhd a\unrhd Y)&=\memt(\Le_a(X))\unlhd a\unrhd\memt(\Ri_a(Y)).
\end{align*}
For $a\in A$, the auxiliary functions $\Le_a: \T\to\T$ and $\Ri_a: \T\to\T$ 
are defined by
\begin{align*}
\Le_a(\tr)&=\tr,\\
\Le_a(\fa)&=\fa,\\
\Le_a(X\unlhd b\unrhd Y)&=
\begin{cases}
\Le_a(X)
&\text{if $b=a$},\\
\Le_a(X)\unlhd b\unrhd \Le_a(Y)
&\text{otherwise},
\end{cases}
\end{align*}
and
\begin{align*}
\Ri_a(\tr)&=\tr,\\
\Ri_a(\fa)&=\fa,\\
\Ri_a(X\unlhd b\unrhd Y)&=
\begin{cases}
\Ri_a(Y)
&\text{if }b= a,\\
\Ri_a(X)\unlhd b\unrhd \Ri_a(Y)
&\text{otherwise}.
\end{cases}
\end{align*}
\end{definition}

As a simple example we depict $se((a\lef b\rig\fa)\lef a\rig\fa)$
and the memorizing evaluation tree $\memse((a\lef b\rig\fa)\lef a\rig\fa)$:
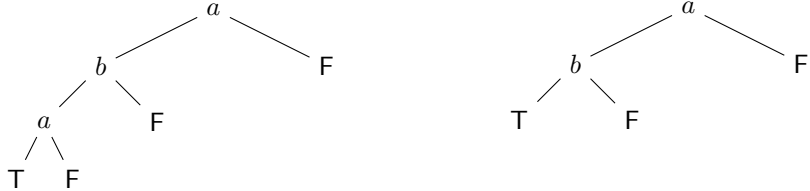
\[
\begin{array}{ll}
\begin{array}{l}
\begin{tikzpicture}[%
level distance=7.5mm,
level 1/.style={sibling distance=30mm},
level 2/.style={sibling distance=15mm},
level 3/.style={sibling distance=7.5mm}
]
\node (a) {$a$}
  child {node (b1) {$b$}
    child {node (c1) {$a$}
      child {node (d1) {$\tr$}} 
      child {node (d2) {$\fa$}}
    }
    child {node (c2) {$\fa$}
    }
  }
  child {node (b2) {$\fa$}
  };
\end{tikzpicture}
\end{array}
&\qquad
\begin{array}{l}
\qquad
\begin{tikzpicture}[%
level distance=7.5mm,
level 1/.style={sibling distance=30mm},
level 2/.style={sibling distance=15mm},
level 3/.style={sibling distance=7.5mm}
]
\node (a) {$a$}
  child {node (b1) {$b$}
    child {node (c1) {$\tr$}
    }
    child {node (c2) {$\fa$}
    }
  }
  child {node (b2) {$\fa$}
  };
\end{tikzpicture}
\\[8mm]
\end{array}\end{array}
\]
The similarities between $\memse$ and the function $\membf$ will of course be exploited.

\begin{lemma}
\label{la:5.11}
For all $P\in\BF$, $\memt(se(P))=se(\mem(P))$.
\end{lemma}

\begin{proof}
We first prove an auxiliary result:
\begin{equation}
\label{aux4}
\text{For all $a\in A$ and $P\in\BF$, $\Le_a(se(P))=se(\ell_a(P))$ and 
$\Ri_a(se(P))=se(\ri_a(P))$.}
\end{equation}
Fix some $a\in A$.
We prove~\eqref{aux4} by structural induction on $P$.
The base cases $P\in\{\tr,\fa\}$ are trivial. 
For the inductive case $P= Q\lef b\rig R$ we have to distinguish
the cases $b=a$ and $b\ne a$. If $b=a$, then 
\begin{align*}
\Le_a(se(Q\lef a\rig R))
&=\Le_a(se(Q)\unlhd a\unrhd se(R))\\
&=\Le_a(se(Q))\\
&=se(\ell_a(Q))
&&\text{by IH}\\
&=se(\ell_a(Q\lef a\rig R)), 
\end{align*}
and if $b\ne a$, then
\begin{align*}
\Le_a(se(Q\lef b\rig R))
&=\Le_a(se(Q)\unlhd b\unrhd se(R))\\
&=\Le_a(se(Q))\unlhd b\unrhd \Le_a(se(R))\\
&=se(\ell_a(Q))\lef b\rig se(\ell_a(R))
&&\text{by IH}\\
&=se(\ell_a(Q\lef b\rig R)).
\end{align*}
The second equality can be derived in a similar way, and
this finishes the proof of~\eqref{aux4}.

\smallskip

We  prove the lemma's statement by induction on $d(P)$. 
The base cases $P\in\{\tr,\fa\}$
follow immediately. Assume $P= Q\lef a\rig R$, then
\begin{align*}
\memt(se(Q\lef a\rig R))
&=\memt(se(Q)\unlhd a\unrhd se(R))\\
&=\memt(\Le_a(se(Q)))\unlhd a\unrhd \memt(\Ri_a(se(R)))\\
&=\memt(se(\ell_a(Q)))\unlhd a\unrhd \memt(se(\ri_a(R)))
&&\text{by \eqref{aux4}}\\
&=se(\mem(\ell_a(Q)))\unlhd a\unrhd se(\mem(\ri_a(R)))
&&\text{by~IH (and Lemma~\ref{la:5.3})}\\
&=se(\mem(\ell_a(Q))\lef a\rig \mem(\ri_a(R)))\\
&=se(\mem(Q\lef a\rig R)).
\end{align*}
\end{proof}

\begin{definition}
\label{def:5.12}
\textbf{Memorizing valuation congruence}, notation $=_\memse$, is defined on \PS\
as follows:
\[P=_\memse Q~\iff~\memse(P)=\memse(Q).\]
\end{definition}

The following characterization result immediately implies that $=_\memse$ is 
indeed a congruence relation on $\PS$.

\begin{proposition}
\label{prop:5.13}
For all $P,Q\in\PS$, $P=_\memse Q~\iff~ P=_\membf Q$.
\end{proposition}

\begin{proof}
($\Rightarrow$) Assume $\memse(P)=\memse(Q)$, thus 
$\memt(se(P))=\memt(se(Q))$.
By Corollary~\ref{cor:1},  
\\
\(\memt(se(\baf(P)))=\memt(se(\baf(Q))),\)
so by Lemma~\ref{la:5.11},
\[se(\mem(\baf(P)))=se(\mem(\baf(Q))).\] 
By Lemma~\ref{la:2.8}, it follows that $\mem(\baf(P))=\mem(\baf(Q))$,
that is, $P=_\membf Q$. 

\smallskip

($\Leftarrow$) If $P=_\membf Q$, then 
$se(\mem(\baf(P)))=se(\mem(\baf(Q)))$,
and by Lemma~\ref{la:5.11}, 
\[\memt(se(\baf(P)))=\memt(se(\baf(Q))).\]
By Corollary~\ref{cor:1}, 
$\memt(se(P))=\memt(se(Q))$, that is, $P=_\memse Q$.
\end{proof}

We end this section with a completeness result for 
memorizing valuation congruence.

\begin{theorem}[Completeness of $\CPmem$]
\label{thm:5.14}
For all $P,Q\in\PS$, 
\[\CPmem\vdash P=Q
~\iff~ 
P=_\memse Q.\]
\end{theorem}

\begin{proof}
Combine Theorem~\ref{thm:5.9} and Proposition~\ref{prop:5.13}. 
\end{proof}

\section{Evaluation trees for static valuation congruence}
\label{sec:stat}
The most identifying axiomatic extension of \CP\ we consider 
can be defined by adding the following axiom to $\CPmem$:
\begin{equation}
\label{eq:Hoare}
\tag{CPs}
\fa\lef x\rig \fa=\fa.
\end{equation}
So, the evaluation value of each atom in a conditional statement is memorized, and
by axiom~\eqref{eq:Hoare}, no  atom $a$ can have a side effect because
$\tr\lef(\fa\lef a\rig\fa)\rig P=\tr\lef\fa\rig P=P$ for all $P\in\PS$. 
We write $\CPstat$ for the set of these axioms, thus
\[
\CPstat=\CPmem\cup\{\eqref{eq:Hoare}\}=\CP\cup\{\eqref{CPmem},\eqref{eq:Hoare}\}.
\]
Observe that 
\(\CPstat\vdash\tr=\tr\lef(\fa\lef x\rig\fa)\rig\tr=
(\tr\lef\fa\rig\tr)\lef x\rig(\tr\lef\fa\rig\tr)=\tr\lef x\rig\tr,
\) 
so the duality principle holds in $\CPstat$.
The following lemma is a direct consequence of axiom~\eqref{eq:Hoare}. 

\begin{lemma}
\label{la:6.1}
For all $P,Q\in\PS$, $\CPstat\vdash P=P\lef Q\rig P$.
\end{lemma}

\begin{proof}
\begin{align*}
\CPstat\vdash P&=\tr\lef(\fa\lef Q\rig\fa)\rig P
&&\text{by~\eqref{eq:Hoare}, \eqref{cp2}}\\
&=(\tr\lef\fa\rig P)\lef Q\rig(\tr\lef\fa\rig P)
&&\text{by~\eqref{cp4}}\\
&=P\lef Q\rig P.
&&\text{by~\eqref{cp2}}
\end{align*}
\end{proof}

Recall that the contraction laws~\eqref{eq:contr}, that is
$x\lef y\rig(v\lef y\rig w)=x\lef y\rig w$,
and~\eqref{eq:contr2}, that is
$(w\lef y\rig v)\lef y\rig x=w\lef y\rig x$, are derivable from $\CPmem$. 
A simple example on $\CPstat$ illustrates how the order of evaluation of 
$x$ and $y$ in $x\lef y\rig\fa$ can be swapped:
\begin{align}
\label{eq:comm}
x\lef y\rig\fa
&=
y\lef x\rig\fa.
\end{align}
Equation~\eqref{eq:comm} can be derived as follows: 
\begin{align*}
\CPstat\vdash x\lef y\rig\fa
&=((\tr\lef x\rig\fa)\lef y\rig\fa)\lef x\rig((\tr\lef x\rig\fa)\lef y\rig\fa)
&&\text{by \eqref{cp3}, Lemma~\ref{la:6.1}}\\
&=(\tr\lef y\rig\fa)\lef x\rig(\fa\lef y\rig\fa)
&&\text{by~\eqref{eq:contr} and~\eqref{eq:contr2}}\\
&=y\lef x\rig\fa.
&&\text{by \eqref{cp3}, \eqref{eq:Hoare}}
\end{align*}

\medskip

In~\cite{BP10} we defined $\CP_{st}$ as the extension of \CP\ with the 
following two axioms:
\begin{align*}
\label{CPstat}\tag{CPstat} 
(x\lef y\rig z)\lef u\rig v
&=(x\lef u\rig v)\lef y\rig (z\lef u\rig v),\qquad\\
\tag{the contraction law~\eqref{eq:contr2}}
(x\lef y\rig z)\lef y\rig u&=x\lef y\rig u.
\end{align*}
Axiom \eqref{CPstat} expresses how the order of evaluation of 
$u$ and $y$ can be swapped.
Because we will rely on results for $\CP_{st}$ recorded in~\cite{BP10},
we first prove the following result.

\begin{proposition}
\label{prop:6.2}
The axiom sets
$\CPst$ and $\CPstat$ are equally strong.
\end{proposition}

\begin{proof} We show that all axioms in the one set are 
derivable from the other set.
We first prove that
the axiom~\eqref{CPmem} is derivable
from $\CPst$:
\begin{align*}
\CPst&\vdash 
x\lef y\rig(z\lef u\rig(v\lef y\rig w))
\\
&=x\lef y\rig((v\lef y\rig w)\lef(\fa\lef u\rig\tr)\rig z)
&&\text{by \eqref{cp4}, \eqref{cp2}, \eqref{cp1}}\\
&=x\lef y\rig
((v\lef(\fa\lef u\rig\tr)\rig z)\lef y\rig
(w\lef(\fa\lef u\rig\tr)\rig z))
&&\text{by \eqref{CPstat}}\\
&= x\lef y\rig(w\lef(\fa\lef u\rig\tr)\rig z)
&&\text{by \eqref{eq:contr}}\\
&= x\lef y\rig(z\lef u\rig w),
&&\text{by \eqref{cp4}, \eqref{cp2}, \eqref{cp1}}
\end{align*}
where the contraction law~\eqref{eq:contr}
is derivable from $\CPst$: replace $y$ by $\fa\lef y\rig\tr$ in~\eqref{eq:contr2}.
Hence $\CPstat\vdash \eqref{CPmem}$.
If $u=v=\fa$ in axiom~\eqref{CPstat},
we find $\fa\lef x\rig \fa=\fa$, hence $\CPst\vdash\CPstat$.

\smallskip

In order to show that $\CPstat\vdash\CPst$ we have to derive
\(\CPstat\vdash \eqref{CPstat}\):
\begin{align*}
\CPstat\vdash~
(x\lef y\rig z)\lef u\rig v
&=(x\lef y\rig (z\lef u\rig v))\lef u\rig v
&&\text{by \eqref{CPmem'}}\\
&=(x\lef y\rig (z\lef u\rig v))\lef u\rig (z\lef u\rig v)
&&\text{by \eqref{eq:contr}}\\
&=x\lef (y\lef u\rig\fa)\rig (z\lef u\rig v)
&&\text{by \eqref{cp4}, \eqref{cp2}}\\
&=x\lef (u\lef y\rig\fa)\rig (z\lef u\rig v)
&&\text{by \eqref{eq:comm}}\\
&=(x\lef u\rig(z\lef u\rig v))\lef y\rig (z\lef u\rig v)
&&\text{by \eqref{cp4}, \eqref{cp2}}\\
&=(x\lef u\rig v)\lef y\rig (z\lef u\rig v).
&&\text{by \eqref{eq:contr}}
\end{align*}
\end{proof}

Given a finite, ordered subset of atoms we
define a proper subset of basic forms with the property that each
propositional statement over these atoms can be proved equal to such a basic form. 

\begin{definition}
\label{def:6.3}
Let $\Au\subset A^\ast$ be the set of strings over $A$ with the property that 
each $\sigma\in \Au$ contains no multiple occurrences of the same atom.\footnote{Recall
  that we write $\epsilon$ for the empty string, thus $\epsilon\in\Au$. If $A$ is finite,
  say $|A|=n$,
  then $|\Au|=a(n)$, where $a(0)=1$ and $a(k+1)=1+(k+1)\cdot a(k)$.}
\textbf{St-basic forms over $\sigma\in\Au$} are defined as follows:
\begin{itemize}\setlength\itemsep{-.1em}
\item
\tr\ and \fa\ are st-basic forms over $\epsilon$, and
\item
$P\lef a\rig Q$ is an st-basic form over $\rho a\in\Au$ if $P$
 and $Q$ are st-basic forms over $\rho$.
\end{itemize}
$P$ is an \textbf{st-basic form} if for some $\sigma\in \Au$, $P$ is an st-basic form
over $\sigma$.
\end{definition}

For example, an st-basic form over $ab\in \Au$ has the following form:
\[(B_1\lef a\rig B_2)\lef b\rig (B_3\lef a\rig B_4)\]
with $B_i\in\{\tr,\fa\}$. For $\sigma=a_1a_2\cdots a_n\in\Au$, 
there exist $2^{2^n}$ different st-basic forms over $\sigma$. 

It will turn out useful to define a function that transforms conditional statements
to st-basic forms. Therefore, given $\sigma\in\Au$
we consider terms in $\PSf{A'}$, where $A'$ is the finite subset of $A$ that contains the
elements of $\sigma$. 

\begin{definition}
\label{def:6.4}
The \textbf{alphabet function} $\ALPHA:\Au\to 2^A$ returns the set of atoms
of a string in $\Au$:
\[
\text{$\ALPHA(\epsilon)=\emptyset$,\quad and \quad 
$\ALPHA(\sigma a)=\ALPHA(\sigma)\cup\{a\}$. }\]
\end{definition}

\begin{definition}
\label{def:6.5}
Let $\sigma\in\Au$. 
The conditional statement $E^\sigma\in\BFf{\ALPHA(\sigma)}$ 
is defined as 
\[
E^\epsilon=\fa\quad
\text{and, if $\sigma=\rho a$,}\quad 
E^\sigma=E^\rho\lef a\rig E^\rho.\]
\end{definition}
So, for each $\sigma\in\Au$, $E^\sigma$ is an st-basic form over $\sigma$
in which the constant
\tr\ does not occur, e.g.,
\[E^{ab}=(\fa\lef a\rig\fa)\lef b\rig(\fa\lef a\rig\fa).\]
\begin{definition}
The \textbf{st-basic form function}
$\stbf_\sigma:\PSf{\ALPHA(\sigma)}\to \PSf{\ALPHA(\sigma)}$ is defined by
\[\stbf_\sigma(P)=\membf(\tr\lef E^\sigma\rig P),\]
where $\membf$ is defined in Definition~\ref{def:5.2}.
\end{definition}

For example, 
$\stbf_{ab}(a)=(\tr\lef a\rig\fa)\lef b\rig(\tr\lef a\rig\fa)$
and
$\stbf_{ba}(a)=(\tr\lef b\rig\tr)\lef a\rig(\fa\lef b\rig\fa)$.

\medskip

The reason that $\stse_\sigma(P)$ is defined relative to some $\sigma\in\Au$ that covers 
the alphabet of $P$ is that in order to prove completeness of $\CPstat$ (and $\CP_{st}$), 
we need to be able to relate conditional statements that contain different sets of atoms, 
but have equal st-basic forms for all appropriate $\sigma$, such as
\[\text{$\stbf_{ba}(\fa)=\stbf_{ba}(\fa\lef a\rig\fa)=\stbf_{ba}(\fa\lef b\rig\fa)
=(\fa\lef b\rig\fa)\lef a\rig(\fa\lef b\rig\fa)$.}
\]

\begin{lemma}
\label{la:6.7}
Let $\sigma\in\Au$.
For all $P\in\PS$,
$\CPstat\vdash P=\tr\lef E^\sigma\rig P$.
\end{lemma}

\begin{proof}
By induction on the structure of $\sigma$.
If $\sigma=\epsilon$, then $E^\sigma=\fa$ and by axiom~\eqref{cp2},
$\CPstat\vdash P=\tr\lef E^\epsilon\rig P$.
If $\sigma=\rho a$ for some $\rho\in\Au$ and $a\in A$, then
$E^\sigma=E^\rho\lef a\rig E^\rho$, and hence
\begin{align*}
\CPstat\vdash P
&=P\lef a\rig P
&&\text{by Lemma~\ref{la:6.1}}\\
&=(\tr\lef E^\rho\rig P)\lef a\rig(\tr\lef E^\rho\rig P)
&&\text{by IH}\\
&=\tr\lef (E^\rho\lef a\rig E^\rho)\rig P.
&&\text{by~\eqref{cp4}}
\end{align*}
\end{proof}

\begin{lemma} 
\label{la:6.8}
Let $\sigma\in\Au$.
For all $P\in\PSf{\ALPHA(\sigma)}$, $\stbf_\sigma(P)$ is an st-basic form.
\end{lemma}

\begin{proof}
We first prove an auxiliary result:
\begin{align}
\label{aux:s2}
&\begin{array}[t]{l}
\text{For all $a\in A$ and $P,Q\in\BF$, }
\ell_a(P[\fa\mapsto Q])=(\ell_a(P))[\fa\mapsto\ell_a(Q)]\\[1pt]
\text{and~}~\ri_a(P[\fa\mapsto Q])=(\ri_a(P))[\fa\mapsto\ri_a(Q)],
\end{array}
\end{align}
Both equalities follow easily by induction on the structure of $P$ and
we only show the inductive case for the first one. Choose $a\in A$.
If $P=P_1\lef a\rig P_2$, then $\ell_a(P)=\ell_a(P_1)$ and
\begin{align*}
&\ell_a(P[\fa\mapsto Q])=\ell_a(P_1[\fa\mapsto Q]\lef a\rig P_2[\fa\mapsto Q])
=\ell_a(P_1[\fa\mapsto Q])\hspace{30mm}\\
&~\stackrel{\text{IH}}=(\ell_a(P_1))[\fa\mapsto\ell_a(Q)]
=(\ell_a(P))[\fa\mapsto\ell_a(Q)],
\end{align*}
and if
$P=P_1\lef b\rig P_2$ with $b\ne a$, then $\ell_a(P)=\ell_a(P_1)\lef b\rig \ell_a(P_2)$ and
\begin{align*}
&\ell_a(P[\fa\mapsto Q])=\ell_a(P_1[\fa\mapsto Q]\lef b\rig P_2[\fa\mapsto Q])
=\ell_a(P_1[\fa\mapsto Q])\lef b\rig\ell_a(P_1[\fa\mapsto Q])\\
&~\stackrel{\text{IH}}=
(\ell_a(P_1))[\fa\mapsto\ell_a(Q)]\lef b\rig(\ell_a(P_2))[\fa\mapsto\ell_a(Q)]
=(\ell_a(P))[\fa\mapsto\ell_a(Q)].
\end{align*}

\smallskip

We prove the lemma's statement by induction on the structure of $\sigma$.
If $\sigma=\epsilon$, then each $P\in\PSf{\ALPHA(\sigma)}$ contains 
no atoms. Hence, $\baf(P)\in\{\tr,\fa\}$.
If $\baf(P)=\tr$ then
\[\stbf_\epsilon(P)=\membf(\tr\lef\fa\rig P)=\mem(\baf(\tr\lef\fa\rig P))
=\mem(\baf(P))=\tr,\]
which is an st-basic form over $\epsilon$. The case for $\baf(P)=\fa$ is similar.

If $\sigma=\rho a$ for some $\rho\in\Au$ and $a\in A$, then for each 
$P\in\PSf{\ALPHA(\sigma)}$,
\begin{align}
\nonumber
&\stbf_\sigma(P)\\
\nonumber
&=\membf(\tr\lef E^\sigma\rig P)\\
\nonumber
&=\mem(\baf(\tr\lef E^\sigma\rig P))\\
\nonumber
&=\mem(E^\sigma[\fa\mapsto\baf(P)])
&&\text{by Lemma~\ref{la:2.14}}\\
\nonumber
&=\mem(E^\rho[\fa\mapsto\baf(P)]\lef a\rig E^\rho[\fa\mapsto\baf(P)])\\
\nonumber
&=\mem(\ell_{a}(E^\rho[\fa\mapsto\baf(P)]))\lef a\rig 
  \mem(\ri_{a}(E^\rho[\fa\mapsto\baf(P)]))
\\
\nonumber
&=\mem(\ell_a(E^\rho)[\fa\mapsto\ell_{a}(\baf(P))])\lef a\rig 
  \mem(\ri_a(E^\rho)[\fa\mapsto\ri_{a}(\baf(P))])
&&\text{by~\eqref{aux:s2}}\\
\nonumber
&=\mem(E^\rho[\fa\mapsto\ell_{a}(\baf(P))])\lef a\rig 
  \mem(E^\rho[\fa\mapsto\ri_{a}(\baf(P))])
&&\text{by $a\not\in\ALPHA(\rho)$}\\
\nonumber
&=\mem(E^\rho[\fa\mapsto\baf(\ell_{a}(\baf(P)))])\lef a\rig 
  \mem(E^\rho[\fa\mapsto\baf(\ri_{a}(\baf(P)))])
&&\text{by Lemma~\ref{la:2.14}}\\
\nonumber
&=\mem(\baf(\tr\lef E^\rho\rig\ell_{a}(\baf(P))))\lef a\rig 
  \mem(\baf(\tr\lef E^\rho\rig\ri_{a}(\baf(P))))
\\
\label{aux:s1}
&=\stbf_\rho(\ell_{a}(\baf(P)))\lef a\rig \stbf_\rho(\ri_{a}(\baf(P))),
&&\text{by~IH}
\end{align}
where~\eqref{aux:s1} follows because 
$\ell_{a}(\baf(P))$ and $\ri_{a}(\baf(P))$ are conditional statements in 
$\PSf{\ALPHA(\rho)}$ (thus, not containing $a$), so by induction,
$\stbf_\rho(\ell_{a}(\baf(P)))$ and $\stbf_\rho(\ri_{a}(\baf(P)))$
are st-basic forms over $\rho$. Hence,
$\stbf_\sigma(P)$ is an st-basic form over $\sigma$.
\end{proof}

With Lemma~\ref{la:6.8} we can easily prove the following result.

\begin{proposition}[{\normalfont $\stbf_\sigma$} is a normalization function]
\label{prop:stat}
Let $\sigma\in \Au$. For each $P\in\PSf{\ALPHA(\sigma)}$, $\stbf_\sigma(P)$ is an st-basic form
over $\sigma$, and
for each st-basic form $P$ over $\sigma$, $\stbf_\sigma(P)= P$.
\end{proposition}

\begin{proof}
The first statement is Lemma~\ref{la:6.8}. 
We prove the second statement by induction on the structure of $\sigma$. 
If $\sigma=\epsilon$, $\stbf_\sigma(P)= P$ by definition.

If $\sigma=\rho a$, then $P=P_1\lef a\rig P_2$
with $P_i$ st-basic forms over $\rho$, thus $\ell_a(P_1)=P_1$ and $\ri_a(P_2)=P_2$. 
For brevity, we identify below $\baf(Q)$ and $Q$ for all $Q\in\BF$:
\begin{align*}
\stbf_\sigma(P)
&=\membf(\tr\lef E^{\rho a}\rig (P_1\lef a\rig P_2))\\
&=\mem(E^{\rho a}[\fa\mapsto P_1\lef a\rig P_2])
&&\text{as above}\\
&=\mem(E^\rho[\fa\mapsto P_1\lef a\rig P_2] \lef a\rig E^\rho[\fa\mapsto P_1\lef a\rig P_2])\\
&=\mem(\ell_a(E^\rho[\fa\mapsto P_1\lef a\rig P_2])\lef a\rig 
  \mem(\ri_a(E^\rho[\fa\mapsto P_1\lef a\rig P_2]))\\
&=\mem(\ell_a(E^\rho)[\fa\mapsto \ell_a(P_1\lef a\rig P_2]))\lef a\rig \mem(\ri_a(E^\rho)
  [\fa\mapsto \ri_a(P_1\lef a\rig P_2]))
&&\text{by~\eqref{aux:s2}}\\
&=\mem(E^\rho[\fa\mapsto P_1])\lef a\rig \mem(E^\rho[\fa\mapsto P_2])
&&\text{by $a\not\in\ALPHA(\rho)$}\\
&=\stbf_\rho(P_1)\lef a\rig \stbf_\rho(P_2)
\\
&=P_1\lef a\rig P_2.
&&\text{by IH}
\end{align*}
\end{proof}

\begin{lemma}
\label{la:6.10}
Let $\sigma\in\Au$. 
For all $P\in\PSf{\ALPHA(\sigma)}$, $\CPstat\vdash P=\stbf_\sigma(P)$.
\end{lemma}

\begin{proof}
By Lemma~\ref{la:6.7}, $\CPstat\vdash P=\tr\lef E^\sigma\rig P$.
By Theorem~\ref{thm:5.7},
$\CPstat\vdash \tr\lef E^\sigma\rig P = 
\\ 
\membf(\tr\lef E^\sigma\rig P)$, 
hence
$\CPstat\vdash P=\stbf_\sigma(P)$.
\end{proof}

\begin{definition}
Let $\sigma\in\Au$. The binary relation $\vcstbf{\sigma}$ on $\PSf{\ALPHA(\sigma)}$
is defined as follows: 
\[
P\vcstbf{\sigma} Q~\iff~\stbf_{\sigma}(P)=\stbf_{\sigma}(Q).
\]
\end{definition}

\begin{theorem}
\label{thm:6.12}
Let $\sigma\in\Au$. For all $P,Q\in\PSf{\ALPHA(\sigma)}$, 
$\CPstat\vdash P=Q~\iff P\vcstbf{\sigma} Q$.
\end{theorem}

\begin{proof}
($\Rightarrow$) Assume $\CPstat\vdash P=Q$. Then, by Lemma~\ref{la:6.10},
$\CPstat\vdash \stbf_\sigma(P)=\stbf_\sigma(Q)$, 
and by Proposition~\ref{prop:6.2},
$\CP_{st}\vdash \stbf_\sigma(P)=\stbf_\sigma(Q)$.
In~\cite{BP10} the following two
statements are proved (Theorem~9.1 and an auxiliary result in its proof),
where $=_{st}$ is a binary relation on $\PS$:
\begin{enumerate}\setlength\itemsep{-.1em}
\item
For all $P,Q\in\PS$,\quad
$\CP_{st}\vdash P=Q ~\iff~ P=_{st} Q$.
\item
For all st-basic forms $P$ and $Q,\quad P=_{st} Q ~\Rightarrow P= Q$.
\end{enumerate}
By Lemma~\ref{la:6.8} 
these statements imply $\stbf_\sigma(P)=\stbf_\sigma(Q)$, and thus $P\vcstbf{\sigma} Q$.

\smallskip

($\Leftarrow$)
Assume $P\vcstbf{\sigma} Q$, thus 
$\tr\lef E^\sigma\rig P=_\membf \tr\lef E^\sigma\rig Q$. By Theorem~\ref{thm:5.9}, 
$\CPmem\vdash \tr\lef E^\sigma\rig P= \tr\lef E^\sigma\rig Q$, and by Lemma~\ref{la:6.7}
this implies $\CPstat\vdash P=Q$.
\end{proof}

Hence, the relation $\vcstbf{\sigma}$ is
a congruence on $\PSf{\ALPHA(\sigma)}$ that is axiomatized by $\CPstat$. 
We define a transformation on evaluation trees that mimics the function
$\stbf_\sigma$ and prove that equality of two such transformed trees characterizes 
the congruence that is axiomatized by $\CPstat$.

\begin{definition}
Let $\sigma\in\Au$. The unary \textbf{static evaluation function} 
\[\stse_\sigma:\PSf{\ALPHA(\sigma)}\to \T\]
yields \textbf{static evaluation trees} and is defined as follows:
\[\stse_\sigma(P)=\memse(\tr\lef E^\sigma\rig P),\]
where $E^\sigma$ is defined in Definition~\ref{def:6.5} and $\memse$ in 
Definition~\ref{def:5.12}.
\end{definition}

As an example, let $P=(a\lef b\rig\fa)\lef a\rig\tr$. We depict 
$se(P)$ at the left-hand side. 
The static evaluation tree $\stse_{ba}(P)$
is depicted in the middle, and the static evaluation tree
$\stse_{ab}(P)$
is depicted at the right-hand side:
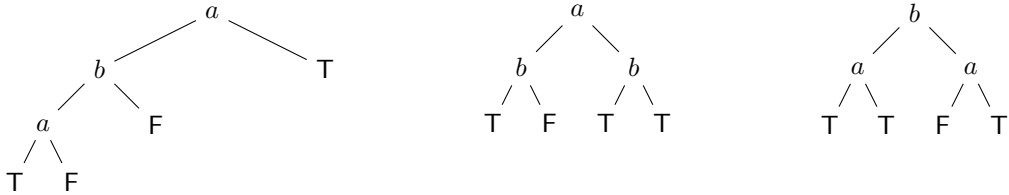
\[
\begin{array}{lll}
\begin{array}{l}
\begin{tikzpicture}[%
level distance=7.5mm,
level 1/.style={sibling distance=30mm},
level 2/.style={sibling distance=15mm},
level 3/.style={sibling distance=7.5mm}
]
\node (a) {$a$}
  child {node (b1) {$b$}
    child {node (c1) {$a$}
      child {node (d1) {$\tr$}} 
      child {node (d2) {$\fa$}}
    }
    child {node (c2) {$\fa$}
    }
  }
  child {node (b2) {$\tr$}
  };
\end{tikzpicture}
\end{array}
&\quad
\begin{array}{l}
\qquad
\begin{tikzpicture}[%
level distance=7.5mm,
level 1/.style={sibling distance=15mm},
level 2/.style={sibling distance=7.5mm},
level 3/.style={sibling distance=7.5mm}
]
\node (a) {$a$}
  child {node (b1) {$b$}
    child {node (c1) {$\tr$}
    }
    child {node (c2) {$\fa$}
    }
  }
  child {node (b2) {$b$}
    child {node (c3) {$\tr$}
    }
    child {node (c4) {$\tr$}
    }
  };
\end{tikzpicture}
\\[8mm]
\end{array}
&\quad
\begin{array}{l}
\qquad
\begin{tikzpicture}[%
level distance=7.5mm,
level 1/.style={sibling distance=15mm},
level 2/.style={sibling distance=7.5mm},
level 3/.style={sibling distance=7.5mm}
]
\node (a) {$b$}
  child {node (b1) {$a$}
    child {node (c1) {$\tr$}
    }
    child {node (c2) {$\tr$}
    }
  }
  child {node (b2) {$a$}
    child {node (c3) {$\fa$}
    }
    child {node (c4) {$\tr$}
    }
  };
\end{tikzpicture}
\\[8mm]
\end{array}
\end{array}
\]
The two different static evaluation trees correspond to the different ways in which one 
can present truth tables for $P$, that is, the different possible
orderings of the valuation values of the atoms occurring in $P$:
\[
\renewcommand*{\arraystretch}{1.2}
\begin{array}{ll|c}
a&b&~(a\lef b\rig\fa)\lef a\rig\tr~\\\hline
\tr&\tr~&\tr\\
\tr&\fa&\fa\\
\fa&\tr&\tr\\
\fa&\fa&\tr
\end{array}
\hspace{2cm}
\begin{array}{ll|c}
b&a&~(a\lef b\rig\fa)\lef a\rig\tr~\\\hline
\tr&\tr~&\tr\\
\tr&\fa&\tr\\
\fa&\tr&\fa\\
\fa&\fa&\tr
\end{array}
\]

The similarities between $\stse_\sigma$ and the function $\stbf_\sigma$ 
can be exploited and lead to our final completeness result.

\begin{definition}
Let $\sigma\in \Au$.
\textbf{Static valuation congruence over $\sigma$}, notation $\vcstse\sigma$, is defined on 
$\PSf{\ALPHA(\sigma)}$ as follows:
\[P\vcstse\sigma Q~\iff~
\stse_{\sigma}(P)=\stse_{\sigma}(Q).
\]
\end{definition}

The following characterization result immediately implies that for all $\sigma\in\Au$,
$\vcstse\sigma$ is indeed a congruence relation on $\PSf{\ALPHA(\sigma)}$.

\begin{proposition}
\label{prop:6.16}
Let $\sigma\in\Au$. For all $P,Q\in\PSf{\ALPHA(\sigma)}$,
\(P\vcstse{\sigma} Q\iff P\vcstbf\sigma Q.
\)
\end{proposition}

\begin{proof}
This follows by Proposition~\ref{prop:5.13}.
\end{proof}

We end this section with a completeness result for 
static valuation congruence.

\begin{theorem}[Completeness of $\CPstat$]
Let $\sigma\in \Au$.
For all $P,Q\in\PSf{\ALPHA(\sigma)}$, 
\[\CPstat\vdash P=Q
~\iff~ 
P\vcstse\sigma Q.\]
\end{theorem}

\begin{proof}
Combine Theorem~\ref{thm:6.12} and Proposition~\ref{prop:6.16}.
\end{proof}

\section{Conclusions}
\label{sec:Conc}

In~\cite{BP10} we introduced proposition 
algebra using Hoare's conditional $x\lef y\rig z$
and the constants $\tr$ and $\fa$.
We defined a number of varieties of so-called
\emph{valuation algebras} in order to capture different semantics for the 
evaluation of conditional statements, and provided axiomatizations for
the resulting valuation congruences: 
$\CP$ (four axioms) characterizes the least identifying valuation congruence
we consider, and the extension $\CPmem$
(one extra axiom) characterizes the
most identifying valuation congruence below ``sequential propositional
logic'' (SPL), while static valuation congruence, axiomatized by adding the simple 
axiom $\fa\lef x\rig\fa=\fa$
to $\CPmem$, can be seen as a characterization of SPL.
In~\cite{HMA,BP12a} we introduced an alternative valuation semantics
for proposition algebra in the form of \emph{Hoare-McCarthy algebras} (HMA's)
that is more elegant than the semantical framework
provided in~\cite{BP10}: HMA-based semantics
has the advantage that one can define a valuation congruence
without first defining the 
valuation \emph{equivalence} it is contained in. 

\medskip

In this paper, we use Staudt's evaluation trees~\cite{Daan} to define free 
valuation congruence as the relation $=_\fr$ (\S\ref{sec:free}), 
and this appears to be a relatively simple and stand-alone exercise, 
resulting in a semantics that is elegant and much simpler than 
HMA-based semantics~\cite{HMA,BP12a} and the semantics defined in~\cite{BP10}.
By Theorem~\ref{thm:1}, $=_\fr$ coincides with ``free valuation congruence as 
defined in~\cite{BP10}''
because both relations are axiomatized by $\CP$ (see \cite[Thm.4.4 and Thm.6.2]{BP10}).
The advantage of ``evaluation tree semantics'' is that for a given conditional statement $P$,
the evaluation tree $se(P)$ determines all relevant evaluations, so $P=_\fr Q$
is determined by evaluation trees that contain no more atoms 
than those that occur in 
$P$ and $Q$; this is comparable to how truth tables can be used in the setting of
propositional logic.

In \S\ref{sec:rp} we define repetition-proof valuation congruence $=_\rpse$
on $\PS$ by $P=_\rpse Q$ if, and only if, $\rpse(P)=\rpse(Q)$, where
$\rpse(P)=\rpt(se(P))$ and $\rpt$ is a transformation function
on evaluation trees. 
It is obvious that
this transformation is ``natural'', given the axiom schemes $\eqref{CPrp1}$ and 
\eqref{CPrp2} that are characteristic for $\CPrp$.
The equivalence on $\PS$ that we want to prove is
\begin{equation}
\label{finaleq}
\text{$\CPrp\vdash P=Q~\iff~ P=_\rpse Q$},
\end{equation}
by which $=_\rpse$ coincides with 
``repetition-proof valuation congruence
as defined in~\cite{BP10}'' because both are axiomatized by $\CPrp$ 
(see \cite[Thm.6.3]{BP10}).
So, by equivalence~\eqref{finaleq}, $=_\rpse$ 
is a \emph{congruence} relation on \PS. However, we could not find a direct proof
of this fact and we chose to simulate the transformation $\rpse$
by the transformation $\rpbf$ on conditional statements,
and to prove that the associated equivalence relation $=_\rpbf$ is 
axiomatized by $\CPrp$, and is hence a congruence. This is Theorem~\ref{thm:3.11},
the proof of which depends on~\cite[Thm.6.3]{BP10}\footnote{%
  This theorem requires that $|A|>1$, and so does the HMA approach in~\cite{HMA}.}
\emph{and} on Theorem~\ref{thm:3.9}, that is,
\[\text{For all $P\in\PS$, $\CPrp\vdash P=\rpbf(P)$}.\]
In order to prove~\eqref{finaleq} (which is Theorem~\ref{thm:3.17}), 
it is thus sufficient to prove
that $=_\rpbf$ and $=_\rpse$ coincide, and this is Proposition~\ref{prop:3.16}.
Although it remains a challenge to find a direct and elegant proof 
of equivalence~\eqref{finaleq},
we can conclude that repetition-proof evaluation trees and the valuation 
congruence $=_\rpse$ provide a full-fledged, simple and elegant semantics for $\CPrp$.

The structure of our completeness proofs of the axiomatizations for the other valuation 
congruences is very similar, although the case for static valuation congruence 
requires a slightly more complex proof (below we return to this point).
Moreover, these axiomatizations are incremental:
the axiom systems $\CPrp$ up to and including  
$\CPstat$ all share the axioms
of \CP, and each succeeding system is defined by 
the addition of either one or two axioms, in most cases making 
previously added axiom(s) redundant. Given some $\sigma\in\Au$,
this implies that in $\PSf{\ALPHA(\sigma)}$,
\[=_\fr~\subseteq~=_\rpse~\subseteq~=_\crse~\subseteq~ =_\memse~\subseteq~ \vcstse\sigma,\]
where all these inclusions are proper.
We conclude that for the valuation congruences $=_\fr$ up to $=_\memse$, 
the associated evaluation trees provide a full-fledged, simple and elegant semantics.

The case for static valuation congruence over $\PSf{\ALPHA(\sigma)}$ for some $\sigma\in\Au$ 
is somewhat more involved.
This semantics coincides with any standard semantics of propositional logic in the following sense: 
\[\text{$P\vcstse\sigma Q$ \quad
if, and only if,
\quad
$\overline P\leftrightarrow \overline Q$~ is a tautology in propositional logic,}\]
where $\overline P$ and $\overline Q$ refer to Hoare's definition~\cite{Hoa85}:
\[\overline{x\lef y\rig z}=
(\overline x\wedge\overline y)\vee(\neg\overline y\wedge\overline z),\quad\overline a=a,\quad 
\overline\tr=\tr,\quad\overline\fa=\fa.\]
Let $\sigma\in\Au$ and $a\in\ALPHA(\sigma)$. 
The fact that $\vcstse \sigma$ identifies more than $=_\memse$ is immediately clear: 
\[\fa\lef a\rig\fa\vcstse{\sigma}\fa,\]
while it is easy to see that $\fa\lef a\rig\fa\ne_\memse\fa$.
Our proof that $\CPstat$ (and thus $\CP_{st}$) is an axiomatization of
static valuation congruence is slightly more complex than those for the other
axiomatizations because in this case the evaluation of a conditional statement $P$ does not enforce
a canonical order for the evaluation of its atoms, and 
therefore such an ordering should be fixed 
beforehand in order to construct an adequate evaluation tree. 
To this purpose, we can use any $\sigma\in\Au$ that covers the atoms in $P$. 

\medskip

A spin-off of our
approach can be called ``normalization functions for proposition algebra'': 
for each valuation congruence C considered, 
two conditional statements are C-valuation congruent if, and only if, 
the basic form function for C returns the same images.\footnote{%
  We use the term ``basic form'' instead of ``normal form'' because  according to \CP,
  the obvious definition of a normal form $t\in\PS$ is that either $t\in A\cup\{\tr,\fa\}$, 
  or $t$ satisfies the property ``if $t_1 \lef t_2\rig t_3$ is a subterm of $t$, then 
  $t_2 \in A$ and it is not the case that $t_1 = \tr$ and $t_3=\fa$''. 
  Also with respect to the other axiomatizations, a single atom $a\in A$ would be a 
  typical normal form.}

\medskip

We conclude with a brief digression on \emph{short-circuit logic}, which we defined 
in~\cite{BPS13} (see~\cite{BP12a} for a quick introduction), and an example on the use
of $\CPrp$.
Familiar binary connectives that 
occur in the context of imperative programming and that prescribe 
short-circuit evaluation, such as 
\texttt{\&\&} (in C called ``logical AND''), are often defined in the following 
way:
\[
P\;\texttt{\&\&}\; Q ~=_{\text{def}}~ \texttt{if }P\texttt{ then }Q\texttt{ else }\false,\]
independent of the precise syntax of $P$ and $Q$,
hence, $P\;\texttt{\&\&}\;Q=_{\text{def}}Q\lef P\rig\fa$.
It easily follows that \texttt{\&\&} is associative (cf.\ Footnote~\ref{fn1}).
In a similarly way, negation can be defined by $\neg P=_{\text{def}}\fa\lef P\rig\tr$.
In~\cite{BPS13} we focus on this question:
\begin{question}
\label{qu:what}
Which are the logical laws that characterize short-circuit 
evaluation of binary propositional connectives?
\end{question}
A first approach to this question is to adopt the conditional as an auxiliary operator, 
as is done in~\cite{BPS13,BP12a}, and to answer Question~\ref{qu:what} using
definitions of the binary propositional connectives as above and the
axiomatization for the valuation congruence of interest in proposition algebra 
(or, if ``mixed conditional statements'' are at stake, axiomatizations for the appropriate
valuations congruences).
An alternative and more direct
approach to Question~\ref{qu:what} is to establish axiomatizations for 
short-circuited binary connectives in which the conditional is \emph{not} used. 
For free valuation congruence, an equational axiomatization of 
short-circuited binary propositional connectives 
is provided by Staudt in~\cite{Daan,PS17}, where 
$se(P\;\texttt{\&\&}\;Q)=_{\text{def}}se(P)[\tr\mapsto se(Q)]$ 
and $se(\neg P)=_{\text{def}}se(P)[\tr\mapsto\fa,\fa \mapsto\tr]$ 
(and where the function $se$ is also defined for short-circuited disjunction), 
and the associated completeness 
proof is based on decomposition properties of such evaluation trees. 
For repetition-proof and contractive valuation congruence we conjecture that a 
finite equational axiomatization of the short-circuited binary propositional 
connectives does not exist if $|A|>2$.  
We conclude with an example on the use of $\CPrp$ that is
based on~\cite[Ex.4]{BPS13}.

\begin{example}\rm
\label{ex:rp2}
Let $A$ be a set of atoms of the form \texttt{($e$==$e'$)} and \texttt{(n=$e$)}
with $\texttt n$ some initialized program variable and $e,e'$ arithmetical expressions over 
the integers that may contain $\texttt n$.
Assume that \texttt{($e$==$e'$)} evaluates to \true\ if $e$ and $e'$
represent the same value, and \texttt{(n=$e$)} always evaluates to 
\emph{true} with the effect that $e$'s value is assigned to $\texttt n$.
Then these atoms satisfy the axioms of $\CPrp$.\footnote{Of
  course, not all equations that are valid in the setting of Example~\ref{ex:rp2}
  follow from $\CPrp$, e.g., $\CPrp\not\vdash\texttt{(0==0)}=\tr$.
  We note that a particular consequence of $\CPrp$ in the setting of short-circuit logic
  is 
  $(\neg a\;\texttt{\&\&}\;a)\;\texttt{\&\&}\;x=\neg a\;\texttt{\&\&}\;a$
  (cf.\ Example~\ref{ex:rp}), and that Example~\ref{ex:rp2} is related to 
  the work of Wortel~\cite{Wortel},
  where an instance of \emph{Propositional Dynamic Logic}~\cite{DynLog1} 
  is investigated in which assignments can be turned into tests;
  the assumption that such tests always evaluate to \emph{true} is natural because 
  the assumption that assignments always succeed is natural.}
Notice that if $\texttt n$ has initial value 0 or 1,
$(\texttt{(n=n+1)}\;\texttt{\&\&}\;\texttt{(n=n+1)})\;\texttt{\&\&}\;\texttt{(n==2)}$ and
$\texttt{(n=n+1)}\;\texttt{\&\&}\;\texttt{(n==2)}$ evaluate to different results,
so the atom \texttt{(n=n+1)} does not satisfy the law $a\;\texttt{\&\&}\;a=a$,
by which this example is typical for the repetition-proof
characteristic of $\CPrp$.
\qedex\end{example}

We finally note that all valuation
congruences considered in this paper can be used as a basis for systematic analysis of 
the kind of \emph{side effects}
that may occur upon the evaluation of short-circuited connectives as in Example~\ref{ex:rp},
and we quote these words of Parnas~\cite{Par10}:
\begin{quote}\it
``Most mainline methods disparage side effects as a bad programming practice. 
Yet even in well-structured, reliable software, many components do have side 
effects; side effects are very useful in practice. It is time to investigate 
methods that deal with side effects as the normal case.''
\end{quote}


\begin{thebibliography}{99}

\bibitem{BBR95}
Bergstra, J.A., Bethke, I., and Rodenburg, P.H. (1995).
A propositional logic with 4 values: true, false, divergent and 
meaningless. 
\newblock \emph{Journal of Applied Non-Classical Logics}, 
5(2):199-218.

\bibitem{BL}
Bergstra, J.A. and Loots, M.E. (2002). 
Program algebra for sequential code.
\newblock \emph{Journal of Logic and Algebraic Programming}, 
51(2):125-156.

\bibitem{HMA}
Bergstra, J.A. and Ponse, A. (2010). 
On Hoare-McCarthy algebras. 
\newblock Available at 
\url{arXiv:1012.5059v1}
[cs.LO].

\bibitem{BP10}
Bergstra, J.A. and Ponse, A. (2011). 
Proposition algebra. 
\newblock \emph{ACM Transactions on Computational Logic}, Vol.~12, No.~3, 
Article 21 (36 pages).

\bibitem{BP12a}
Bergstra, J.A. and Ponse, A. (2012). 
Proposition algebra and short-circuit logic. 
\newblock In Arbab,~F. and Sirjani, M. (eds.), 
\emph{Proceedings of the 4th International Conference on Fundamentals 
of Software Engineering} (FSEN 2011), Tehran.
Volume 7141 of Lecture Notes in Computer Science, 
pages 15-31. Springer.

\bibitem{BP15}
Bergstra, J.A. and Ponse, A. (2015).
Evaluation trees for proposition algebra: 
The case for free and repetition-proof valuation congruence. 
In: Meyer, R., Platzer, A., and Wehrheim,~H. (eds.), 
\emph{Correct System Design (Olderog-Festschrift)}. Volume 9360
of Lecture Notes in Computer Science, pages 44-61, Springer. 

\bibitem{BPS13}
Bergstra, J.A., Ponse A., and Staudt, D.J.C. (2013). 
Short-circuit logic. 
\newblock Available at \url{arXiv:1010.3674v4} [cs.LO,math.LO], 
18 Oct 2010; this version (v4): 12 Mar 2013.

\bibitem{DynLog1}
Harel, D. (1984). 
Dynamic logic. 
\newblock In: Gabbay, D. and G\"unthner, F. (eds.), 
\emph{Handbook of Philosophical Logic}, Volume II, pages 497-604.
Reidel Publishing Company.

\bibitem{HHH87}
Hayes, I.J., 
He Jifeng,
Hoare, C.A.R., Morgan, C.C.,
Roscoe, A.W., Sanders, J.W., Sorensen, I.H.,
Spivey, J.M., and Sufrin B.A. (1987).
Laws of programming.
\newblock {\em Communications of the ACM}, 3(8):672-686.

\bibitem{Hoa85a}
Hoare, C.A.R. (1985).
{\em Communicating Sequential Processes}.
\newblock Prentice-Hall, Englewood Cliffs.

\bibitem{Hoa85}
Hoare, C.A.R. (1985).
A couple of novelties in the propositional calculus. 
\newblock \emph{Zeitschrift f\"ur Mathematische Logik und 
Grundlagen der Mathematik}, 31(2):173-178.
\\
Republished in: Hoare, C.A.R. and Jones, C.B. (1989).
\emph{Essays in Computing Science}, pages 325-331.
\newblock Prentice-Hall, Englewood Cliffs.

\bibitem{Par10}
Parnas, D.L. (2010).
Really Rethinking `Formal Methods'.
\newblock \emph{Computer}, 43(1):28-34, IEEE Computer Society (Jan. 2010).

\bibitem{PS17}
Ponse A. and Staudt, D.J.C. (2017). 
An independent axiomatization for free short-circuit logic. 
\newblock Available at \url{arXiv:1707.05718v1} [cs.LO], 
17 July 2017.

\bibitem{Daan}
Staudt, D.J.C. (2012).
Completeness for two left-sequential logics.
\newblock MSc. thesis Logic, University of Amsterdam
(May 2012).
Available at \url{arXiv:1206.1936v1} [cs.LO].

\bibitem{Wortel}
Wortel, L. (2011).
Side effects in steering fragments.
\newblock MSc. thesis Logic, University of Amsterdam
(September~2011).
Available at \url{arXiv:1109.2222v1} [cs.LO].

\end{thebibliography}
\end{document}